\documentclass[journal]{IEEEtran}
\IEEEoverridecommandlockouts
\usepackage[utf8]{inputenc}
\usepackage{colortbl} 
\usepackage[T1]{fontenc}
\usepackage{textcomp}
\usepackage{fixltx2e}
\usepackage{microtype}
\usepackage{calc}
\usepackage[normalem]{ulem}
\usepackage{makecell}
\usepackage{diagbox}
\usepackage{pifont}
\usepackage{gensymb}
\usepackage{comment}
\usepackage{bbm}
\usepackage{amsthm}
\newtheorem{Theorem}{Theorem}
\newtheorem{Proposition}{Proposition}

\usepackage{amsmath,amssymb,amsfonts}
\usepackage{algorithmic}
\usepackage{multirow}
\usepackage{textcomp}
\usepackage{xcolor}\usepackage[dvipsnames]{xcolor}
\usepackage{url}
\usepackage{ tipa }
\usepackage[range-phrase=--,per-mode=symbol-or-fraction,binary-units=true,range-units=single,list-units=single,detect-all]{siunitx}
\usepackage[capitalize,noabbrev]{cleveref}
\crefformat{footnote}{#2\footnotemark[#1]#3}

\usepackage{booktabs}
\usepackage{url}
\RequirePackage{xstring}
\RequirePackage{xparse}
\RequirePackage[]{acro}
\NewDocumentCommand\acrodef{mO{#1}mG{}}{\DeclareAcronym{#1}{short={#2}, long={#3}, #4}}
\usepackage[ruled,linesnumbered]{algorithm2e}
\usepackage{siunitx}

\DeclareUnicodeCharacter{0301}{\'{e}}
\acrodef{mmWave}{millimeter-wave}
\acrodef{ISAC}{integrated sensing and communications}
\acrodef{PMCW}{phase-modulated continuous waveform}
\acrodef{CO-PMCW}{code-orthogonal PMCW}
\acrodef{CSO-PMCW}{code-space-orthogonal PMCW}
\acrodef{FMCW}{frequency-modulated continuous waveform}
\acrodef{OFDM}{orthogonal frequency-division multiplexing}
\acrodef{MIMO}{multiple-input multiple-output}
\acrodef{V2X}{vehicle-to-everything}
\acrodef{PAPR}{peak-to-average power ratio}
\acrodef{SINR}{signal-to-interference-plus-noise ratio}
\acrodef{RDM}{range-Doppler map}
\acrodef{SISO}{single-input single-output}
\acrodef{CRLB}{Cramer-Rao lower bound}
\acrodef{ULA}{uniform linear array}
\acrodef{UL}{uplink}
\acrodef{UE}{user equipment}
\acrodef{PRNS}{pseudo-random noise sequence}
\acrodef{AoD}{angle of departure}
\acrodef{AoA}{angle of arrival}
\acrodef{AWGN}{additive white Gaussian noise}
\acrodef{RCS}{radar cross section}
\acrodef{RSU}{roadside unit}
\acrodef{LoS}{line-of-sight}
\acrodef{NLoS}{non-line-of-sight}
\acrodef{FFT}{fast Fourier transform}
\acrodef{IFFT}{inverse fast Fourier transform}
\acrodef{PRI}{pulse repetition interval}
\acrodef{i.i.d.}{independent and identically distributed}
\acrodef{MU-MIMO}{multi-user MIMO}
\acrodef{FIM}{Fisher information matrix}
\acrodef{IIoT}{intelligent internet of things}
\acrodef{CU}{central unit}
\acrodef{LTE}{long-term evolution}
\acrodef{NR}{new radio}
\acrodef{C-RAN}{cloud radio access network}
\acrodef{RRA}{radio resource allocation}
\acrodef{MSE}{mean squared error}
\acrodef{VA}{virtual array}
\acrodef{EFIM}{equivalent Fisher information matrix}
\acrodef{TDM}{time-division multiplexing}
\acrodef{RMS}{root mean square}
\acrodef{FDM}{frequency-division multiplexing}
\acrodef{CDM}{code-division multiplexing}
\acrodef{OTFS}{orthogonal time frequency space}
\acrodef{PEB}{position error bound}
\acrodef{VEB}{velocity error bound}
\acrodef{MMS}{multi-monostatic sensing}
\acrodef{MBS}{multi-bistatic sensing}
\acrodef{MXS}{multi-X-static hybrid sensing}
\acrodef{UAV}{unmanned aerial vehicle}
\acrodef{BS}{base station}
\acrodef{SNR}{signal-to-noise ratio}
\acrodef{eMBB}{enhanced mobile broadband}
\acrodef{MUSIC}{multiple signal classification}
\acrodef{CS}{compressive sensing}
\acrodef{FE}{fine estimation}
\acrodef{Tx}{transmitter}
\acrodef{Rx}{receiver}
\acrodef{TRx}{transceiver}
\acrodef{DL}{downlink}
\acrodef{UP}{uplink}
\acrodef{DoF}{degree of freedom}
\acrodef{ESNR}{energy SNR}
\acrodef{BPSK}{binary phase shift keying}
\acrodef{ICI}{inter-carrier interference}
\acrodef{CP}{cyclic prefix}
\acrodef{ISFFT}{inverse symplectic finite Fourier transform}
\acrodef{SFFT}{symplectic finite Fourier transform}
\acrodef{DD}{delay-Doppler}
\acrodef{RU}{radio unit}
\acrodef{SCA}{successive convex approximation}
\acrodef{SDP}{semidefinite programming}
\acrodef{PSD}{positive semidefinite}
\acrodef{ISD}{inter-side distance}
\acrodef{MRT}{maximum ratio transmission}
\acrodef{NMAE}{normalized mean absolute error}

\usepackage[caption=false,font=footnotesize]{subfig}
\usepackage[pdftex]{graphicx}
\DeclareGraphicsExtensions{.pdf,.png,.jpg,.tikz}
\usepackage{csquotes}
\usepackage[backend=biber,style=ieee,doi=false,isbn=false,sorting=none,sortcites=true,mincitenames=1,maxcitenames=2]{biblatex}
\DeclareFieldFormat{sentencecase}{#1} 
\DeclareFieldFormat{titlecase}{#1} 
\usepackage{xpatch}
\xpatchbibmacro{textcite}{\addspace}{\addnbspace}{}{}
\xpatchbibmacro{Textcite}{\addspace}{\addnbspace}{}{}
\DefineBibliographyStrings{english}{
	andothers = et~al\adddot\addspace
}

\graphicspath{{./Bilder/}}    
\begin{document}

\title{Fundamental Sensing Limits of 6G Cooperative MIMO-ISAC Networks: Joint Position–Velocity CRLB and Decoupling Analysis}

\author{%
Yanpeng Su,~\IEEEmembership{Graduate Student Member,~IEEE}, Norman Franchi,~\IEEEmembership{Member,~IEEE},\\and Maximilian Lübke,~\IEEEmembership{Member,~IEEE}
\thanks{This work was supported by BMFTR, Germany, under Open6GHub+ (16KIS2404). The work contributes to the research within the 6G-Valley innovation cluster. \textit{(Corresponding author: Yanpeng Su.)}}
\thanks{Yanpeng Su, Norman Franchi, and Maximilian Lübke are with the Institute for Smart Electronics and Systems, Friedrich-Alexander-Universität Erlangen-Nürnberg, Erlangen 91058, Germany (email: yanpeng.su@fau.de; norman.franchi@fau.de; maximilian.luebke@fau.de).}
\thanks{Color versions of one or more of the figures in this article are available online at http://ieeexplore.ieee.org.}
}
\markboth{IEEE tbd}{Su \textit{et al.}: Fundamental Sensing Limits of 6G Cooperative MIMO-ISAC Networks: Joint Position–Velocity CRLB and Decoupling Analysis}

\maketitle

\begin{abstract}
This paper presents a Cramér–Rao lower bound (CRLB)-based performance bound analysis of cooperative multiple-input multiple-output (MIMO) integrated sensing and communications (ISAC) networks. 
We first show the CRLB transformation of the signal-level parameters to the state parameters (position and velocity) in cooperative ISAC networks.
Unlike existing studies that primarily ignored coupling between position and velocity in the Fisher information matrix (FIM), we derive the full FIM and the corresponding exact CRLB.
Particularly, the results of multi-monostatic sensing, multi-bistatic sensing, and their hybrid are discussed. 
Addressing the complexity and tractability, we simplify the FIM and CRLB by excluding the coupling terms between the position and velocity, and provide a criterion for determining whether the simplification is valid.
The simplified CRLB benefits from low computational complexity and provides a tractable and reliable performance metric for optimization problems such as resource allocation and beamforming. 
Finally, the position and velocity CRLBs and the simplification-induced error are examined in the simulation.
The results demonstrate that the simplified CRLB can be applied in general cases. Based on the simulation results, the impact of resource and geometric parameters on position and velocity error bounds, and the validity of the simplified CRLBs is explained through the corresponding CRLB expressions.
\end{abstract}

\vspace{-1em}

\section{Introduction}
\Ac{ISAC} is regarded as a key technology of the upcoming 6G mobile networks since it allows the networks to be aware of the environments without significantly increasing resource conflicts, energy consumption, and hardware costs \cite{9924202,9737357}, accelerating the applications of intelligent transportation systems, smart home and factory, and \ac{IIoT} \cite{10536135,persp}. Although the 4G \ac{LTE} and 5G \ac{NR} networks already provide localization services for active devices like smartphones \cite{9354629}, the above-mentioned applications require future 6G mobile networks to be able to perceive both active and passive objects \cite{9591277}. 

Since the ages of 4G and 5G, cooperative transmission and reception techniques have been widely deployed to suppress inter-cell interference at cell edges and improve channel capacity \cite{9272226,3gpp.36.819}, while the \ac{C-RAN} architecture enables efficient coordination and information sharing across distributed nodes, facilitating cooperative operation in multi-node networks \cite{9762865}. At the age of 6G, cooperative sensing is expected to be integrated into the existing communication network, providing higher spatial diversity and resolution and enabling 2D/3D velocity estimation \cite{9585321}. Therefore, the cooperative ISAC networks have attracted significant research attention in recent years \cite{10599241,9411178}. 

Performance bound analysis is one of the cornerstones for the deployment of cooperative ISAC in mobile networks, since it guides the parameter configuration and geometrical deployment of the network while providing fundamental metrics for system optimization designs, such as beamforming \cite{10577579} and \ac{RRA} \cite{9945983}. 
In this context, \ac{CRLB} is the most commonly addressed sensing metric.
CRLB calculates the theoretical lower bound on the \ac{MSE} of any unbiased estimator \cite{1420803}, thereby can be 
used to evaluate the best achievable accuracy of the signal-level parameters, including amplitude, phase, range (delay), radial velocity (Doppler), and \ac{AoA}, and the state parameters, including position and velocity. 

An accurate CRLB analysis requires joint analysis of all unknown parameters since the coupling between the parameters in the \ac{FIM} may influence the result \cite{923295}. 
It is demonstrated in our previous work \cite{11570946} that for common ISAC waveforms, the signal-level parameters of interest (range, radial velocity, and AoA) are completely or approximately decoupled from each other in the \ac{EFIM}. 
Building on our previous results, this work extends the system to cooperative ISAC networks and provides a transformation from signal-level parameters to state parameters, where the coupling between the position and velocity EFIMs is addressed.
\vspace{-1em}

\subsection{Related Works}
The cooperative ISAC networks put forward requirements for the performance analysis of multi-node sensing, where the global FIM is calculated by the sum of the local ones corresponding to the individual nodes. 
Studies in \cite{7131233,10615966,10646234,10419729,10693955,11202824,10266619} focus on the calculation of \ac{PEB}, while its coupling with \ac{VEB} is usually not considered. The study in \cite{7131233} discusses the PEB of a single-antenna radar network under the assumption of static targets, where the PEB is directly determined by the range CRLB. 
In contrast, \cite{10615966} provides a CRLB of narrowband PEB purely based on the angle information. 
The study in \cite{10646234} introduces a concept of error ellipsoid that expresses the error in different directions, and the PEB in realistic scenarios is evaluated.
The PEB of \ac{MIMO} \ac{OFDM} radar network is investigated in \cite{10419729,10693955}, focusing on the \ac{MBS} and \ac{MMS} systems, respectively. 
The work in \cite{11202824} focuses on the PEB of \ac{MBS} OFDM radar systems and particularly investigates the impact of network geometry, highlighting the importance of uniform node distribution. The study in \cite{10266619} focuses on 1-D positioning in a single direction, where the road is abstracted as a line, and the location of the vehicles is estimated. On the contrary to the above-listed works, the VEB of OFDM radar with widely separated antennas is addressed in \cite{5393291}, where the coupling with PEB is also neglected. 

Studies on joint position and velocity analysis remain limited and have attracted increasing attention in recent years. 
The study in \cite{7362229} implemented a joint position and velocity CRLB analysis for radar networks with single-antenna nodes and provided a full expression for PEB and VEB, taking the impact of coupling into account while omitting angle information. Recent work in \cite{11481148} also analyzes PEB and VEB jointly in a single-antenna radar network, and the results show that the impact of varying \ac{RCS} on PEB and VEB is negligible. 
Studies in OFDM radar networks can be found in \cite{11231051,pucci2025,xu2026}. The work in \cite{11231051} focuses on the CRLBs of position and absolute velocity of \ac{MMS} in the detection of \acp{UAV}, under the assumption that the moving direction is known, whereas the PEB is calculated without considering the coupling with velocity. The PEB and VEB of MMS, MBS, and their hybrid are addressed in \cite{pucci2025}, whereas the PEB is calculated in the same manner as \cite{11231051}, and the VEB is obtained by summing the local velocity EFIMs. One problem is that, unlike FIM, EFIM does not obey the law of accumulation, and the produced error in CRLB is not evaluated. The work in \cite{xu2026} models the cooperation between a single \ac{BS} and a single \ac{UE}, forming a single monostatic plus single bistatic sensing system, while the position-velocity coupling is not considered.
\vspace{-1em}


\subsection{Contributions}

Related works exhibit limitations in handling multi-antenna nodes or in capturing the coupling between position and velocity.
As future 6G mobile networks and ISAC systems will inherently adopt massive \ac{MIMO} technology \cite{8869705,9737357}, studies on the CRLB of networked ISAC should account for multi-antenna nodes, where the antenna array introduces AoA-related components in the FIM. Moreover, the full and rigorous CRLB characterizing the coupling between position and velocity is needed, since it provides exact guidance for system design and node distribution.
In contrast, optimization problems generally desire a tractable and simple CRLB metric to avoid high complexity, while maintaining a small deviation from the exact CRLB. Therefore, it is necessary to implement a theoretically justified simplification based on the full CRLB and evaluate the caused deviation.

This paper addresses these gaps by providing a comprehensive CRLB analysis for state parameters. We consider the \ac{DL} sensing between the network nodes in mobile networks, where the \ac{eMBB} and massive MIMO technologies result in multi-antenna radar networks with wideband signals, providing a much higher spatial resolution and transmission power than \ac{UL} sensing while avoiding synchronization and privacy issues \cite{9585321}. In addition, considering the existence of various waveform candidates for sensing \cite{9627227}, such as \ac{FMCW}, \ac{PMCW}, \ac{OFDM}, and \ac{OTFS}, the analysis in this work is based on the generic signal model within the unified framework \cite{11570946}. Therefore, the conclusion can be generalized to different waveforms by substituting their characteristics.
The main contributions of this paper are summarized as follows:
\begin{itemize}
    \item \textbf{Joint PEB and VEB analysis for networked ISAC:} Based on the FIM of signal-level parameters, we provide a transformation to the FIM of state parameters and derive the full PEB and VEB, taking the coupling between them and MIMO node structure into account. The influence of the signal-level parameters and system geometric characteristics on the PEB and VEB is explained.
    \item \textbf{Simplified CRLBs and validity criterion:} The simplified CRLBs are derived by excluding the coupling terms between position and velocity in the FIM. We prove that in a network that can provide sufficient spatial diversity, for each sensing link, if the target's normalized range is much larger than the normalized velocity, the gap between the full and the simplified CRLBs can be neglected. The practicality of this criterion under practical 5G system configurations and scenarios is described. This criterion provides a practical guideline for determining whether the simplified CRLB can be used reliably in system designs.
    \item \textbf{Various sensing types:} The analysis of full and simplified PEB and VEB is implemented for different sensing types: \ac{MMS}, \ac{MBS}, and their hybrid, named \ac{MXS} in our previous work \cite{10599241}. We show the consistency in their CRLB expressions.
    \item \textbf{Verification in Simulation:} The coverage and accuracy of different sensing types, as well as the gap between the full and simplified CRLBs, are investigated geometrically via the CRLB heatmap. Particularly, the feasibility and tightness of the proposed criterion are evaluated. It is demonstrated that the simplified CRLB can be safely used for the area inside the polygon formed by the nodes, which matches the prediction of the criterion. Based on the results, the impact of resource and geometric parameters is described with reference to the CRLB formulas.
\end{itemize}

The rest of this paper is organized as follows: Section~\ref{sec.model} introduces the networked ISAC system model. The full CRLB is provided in Section~\ref{sec.full}. The simplified CRLB and the proposed criterion are presented in Section~\ref{sec.simp}. The simulation results are given in Section~\ref{sec.sim}. Section~\ref{sec.dis} further discusses the results in combination with the theoretical analysis. Section~\ref{sec.con} concludes this work.

\textit{Notations:}  $(\cdot)^{*}$, $(\cdot)^{T}$, $(\cdot)^{H}$, $(\cdot)^{-1}$, and $(\cdot)^{\dagger}$ denote the conjugate, transpose, conjugate transpose, inverse, and Moore–Penrose pseudoinverse, respectively. $a$ and $A$ represent scalars, $\mathbf{a}$ and $\mathbf{A}$ represent vectors and matrices, respectively. 
$\text{tr}(\cdot)$ denotes the trace of a square matrix.
$\mathcal{R}\{\cdot\}$ implies the real-part operator.
$\mathcal{CN}(\mu,\sigma^2)$ denotes the complex normal distribution with mean and variance of $\mu$ and $\sigma^2$. 
\vspace{-1em}

\section{System Model}\label{sec.model}
This work focuses on the sensing functionality in traffic scenarios, where the planar position and velocity are involved. 
Fig.~\ref{fig:model} illustrates an example of the cooperative network, where the nodes can operate in MMS, MBS, or MXS modes.
The positions of the $n$-th node and target are represented in Cartesian coordinates $(x_n,y_n)$ and $(x,y)$, and the velocity of the target is represented by $(v_x,v_y)$.
We assume all the \ac{Tx} and \ac{Rx} antennas are \acp{ULA} with sizes of $N_\text{t}$ and $N_\text{r}$, respectively. The inter-element spacing is set to $d=\lambda/2$, where $\lambda$ denotes the wavelength.
The \ac{AoD} $\vartheta$ and \ac{AoA} $\theta$ denote the angle between the target and the antenna normal, and $\varphi$ denotes the bearing angle, i.e., the angle between the target direction and the $x$-axis. In this paper, the number of nodes and sensing links is denoted by $N$ and $M$, respectively.

\begin{figure}
    \centering
    \includegraphics[width=0.7\linewidth]{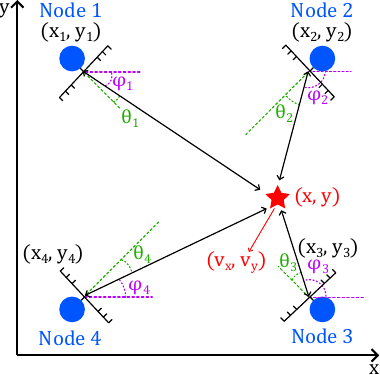}
    \caption{An example of the considered system model. All the sensing nodes can perform transmission or reception, or both.}
    \label{fig:model}
    \vspace{-1em}
\end{figure}

Considering in MBS and MXS, one Tx or Rx may correspond to several observations, we model the received signals by sensing links in this work.
The received signal of the $m$-th link corresponding to the $(n_\text{t},n_\text{r})$-th Tx and Rx is given by
\begin{align}
    \mathbf{y}_{m}(t)
    =\underbrace{A_{m}e^{j\phi_{m}}\mathbf{a}_{m}^*(\theta_{m}){s}_{m}(t-\tau_{m})e^{j2\pi f_{m}^\text{D}t}}_{\pmb{\mu}_{m}(t)}+\mathbf{n}_{m}(t),\label{eq.rec}
\end{align}
where $A_{m}$, $\phi_{m}$, $\tau_{m}$, and $f^\text{D}_{m}$ denote the amplitude, phase, delay, and Doppler shift of the $m$-th link. $\mathbf{a}_{m}^*(\theta_{m})=\mathbf{a}_{n_\text{r}}^*(\theta_{n_\text{r}})$ denotes the steering vector of the $m$-th link Rx. $A_{m}=a_{m}\mathbf{b}_{m}^H(\vartheta_{m})\mathbf{w}_m=a_{m}\mathbf{b}_{n_\text{t}}^H(\vartheta_{n_\text{t}})\mathbf{w}_{n_\text{t}}=N_\text{t}a_{m}$ combines the path loss and transmission gain, 
where $\mathbf{b}_{m}^H(\vartheta_{m})$ is the Tx steering vector, $\mathbf{w}$ is the beamforming vector. The \ac{AoD} $\vartheta_{m}$ is not considered in FIM since it does not introduce new \ac{DoF} of independent observation as $\mathbf{b}_{m}^H(\vartheta_{m})\mathbf{w}$ is scalar and cannot lead to additional observation dimensions \cite{11570946}. ${s}_{m}(t)={s}_{n_\text{t}}(t)$, where $\mathbb
E\{{s}_{n_\text{t}}(t){s}_{n_\text{t}'}(t)\}=0,\forall n_\text{t}\ne n_\text{t}'$ for ease of separation at the Rx \cite{4408448}.
$\mathbf{n}_{m}(t)=\mathbf{n}_{n_\text{r}}(t)\sim\mathcal{CN}(0,\sigma^2)^{N_\text{r}\times1}$ denotes the \ac{AWGN} whose power is assumed to be the same for all the links.

In our previous work \cite{11570946}, 
we demonstrated that for common ISAC waveforms, like PMCW, OFDM, and OTFS, the \ac{EFIM} of the signal-level parameters of interest (range, radial velocity, AoA) $\pmb{\Theta}=[r,v_{\text{r}},\theta]^T$ is completely or approximately diagonal:
\begin{align}\label{eq.sige}
    \mathbf{E}\negthinspace=\negthinspace\mathbf{C}^{-1}\negthinspace\approx\negthinspace\begin{bmatrix}
        E_r & 0 & 0\\
        0 & E_{v_\text{r}} & 0\\
        0 & 0 & E_{\theta}
    \end{bmatrix}\negthinspace=\negthinspace\begin{bmatrix}
        1/C_r & 0 & 0\\
        0 & 1/C_{v_\text{r}} & 0\\
        0 & 0 & 1/C_{\theta}
    \end{bmatrix}\negthinspace,\negthinspace
\end{align}
where for monostatic radar, the relationship between the EFIM elements of range and delay, as well as radial velocity and Doppler, follows $E_{r}^\text{M}=\big(\frac{2}{c_0}\big)^2E_{\tau},  E_{v_\text{r}}^\text{M}=\big(\frac{2}{\lambda}\big)^2E_{f^\text{D}}$; for bistatic radar, $E_{r}^\text{B}=\frac{1}{c_0^2}E_{\tau},  E_{v_\text{r}}^\text{B}=\frac{1}{\lambda^2}E_{f^\text{D}}$, and
\begin{align}\label{eq.siglvcrlb}
    &E_\tau={8\pi^2N_\text{r}\gamma B_\text{rms}^2},\quad E_{f^\text{D}}={8\pi^2N_\text{r}\gamma T_\text{rms}^2},\nonumber\\ &E_{\theta}=\frac{\pi^2\cos^2(\theta)N_\text{r}(N_\text{r}^2-1)\gamma}{6},
\end{align}
where $\gamma$ denotes the \ac{ESNR} calculating the overall energy ratio between signal and noise.
For signals with uniform power distribution over bandwidth and frame, such as FMCW, sinc-shaped PMCW, OFDM, and OTFS, the \ac{RMS} bandwidth and time follow $B_\text{rms}^2=B^2/12$, $T_\text{rms}^2=T_\text{F}^2/12$, where $B$ and $T_\text{F}$ are the bandwidth and frame length.

Independent of sensing mode, the relationship between the AoA and the bearing angle of the $m$-th link Rx is given by
\begin{align}
    \theta_m=\varphi_m^\text{r}-\theta_{0,m},
\end{align}
where $\theta_{0,m}$ represents the antenna normal direction of the Rx, $\varphi^\text{r}_m$ only depends on the Rx antenna. Hence, the formulas of AoA CRLB $C_{\theta}$ for monostatic and bistatic radars are the same.
\vspace{-1em}

\section{Full CRLB Derivation}\label{sec.full}
The transformation from signal-level EFIM to the \ac{FIM} of state parameters $\pmb{\eta}=[x,y,v_x,v_y]^T$ follows the chain rule since $\big(\frac{\partial \pmb{\mu}(t)}{\partial\pmb{\eta}}\big)^H\frac{\partial \pmb{\mu}(t)}{\partial\pmb{\eta}}=\big(\frac{\partial \pmb{\mu}(t)}{\partial\pmb{\Theta}}\frac{\partial\pmb{\Theta}}{\partial\pmb{\eta}}\big)^H\frac{\partial \pmb{\mu}(t)}{\partial\pmb{\Theta}}\frac{\partial\pmb{\Theta}}{\partial\pmb{\eta}}$. Consequently, the state parameter FIM of the $m$-th link is calculated by
\begin{align}
    \mathbf{F}_m\negthinspace=\negthinspace\frac{2}{\sigma^2}\mathcal{R}\Big\{\negthinspace\int_{-\infty}^{\infty}\negthinspace\Big(\frac{\partial \pmb{\mu}_m(t)}{\partial\pmb{\eta}}\Big)^{\negthinspace H}\frac{\partial \pmb{\mu}_m(t)}{\partial\pmb{\eta}}dt\Big\}
    \negthinspace=\negthinspace\mathbf{J}^{ T}_m\mathbf{E}_m\mathbf{J}_m,
\end{align}
where $\mathbf{J}_m$ is the Jacobian matrix of parameter transformation:
\begin{align}
    \mathbf{J}_m=\frac{\partial\pmb{\Theta}_m}{\partial\pmb{\eta}}=\begin{bmatrix}
        \frac{\partial r_m}{\partial x} & \frac{\partial r_m}{\partial y} & \frac{\partial r_m}{\partial v_x} & \frac{\partial r_m}{\partial v_y}\\
        \frac{\partial v_{\text{r},m}}{\partial x} & \frac{\partial v_{\text{r},m}}{\partial y} & \frac{\partial v_{\text{r},m}}{\partial v_x} & \frac{\partial v_{\text{r},m}}{\partial v_y}\\
        \frac{\partial \theta_m}{\partial x} & \frac{\partial \theta_m}{\partial y} & \frac{\partial \theta_m}{\partial v_x} & \frac{\partial \theta_m}{\partial v_y}
    \end{bmatrix}.
\end{align}

For independent observations of different links, the global FIM is calculated by the sum of the local FIMs:
\begin{align}\label{eq.FIM}{\setlength{\arraycolsep}{1.6pt} 
    \negthinspace\mathbf{F}\negthinspace=\negthinspace\sum_{m=1}^M\negthinspace\mathbf{J}_m^T\mathbf{E}_m\mathbf{J}_m\negthinspace=\negthinspace\begin{bmatrix}
        \mathbf{F}_\text{P} & \mathbf{F}_\text{PV}\\
        \mathbf{F}_\text{PV}^T & \mathbf{F}_\text{V}
    \end{bmatrix}\negthinspace=\negthinspace\begin{bmatrix}
        F_{xx} & F_{xy} & F_{xv_x} & F_{xv_y}\\
        F_{xy} & F_{yy} & F_{yv_x} & F_{yv_y}\\
        F_{xv_x} & F_{yv_x} & F_{v_xv_x} & F_{v_xv_y}\\
        F_{xv_y} & F_{yv_y} & F_{v_xv_y} & F_{v_yv_y}
    \end{bmatrix}\negthinspace\negthinspace,\negthinspace
}\end{align}
where $\mathbf{F}_\text{P}$ and $\mathbf{F}_\text{V}$ are the position and velocity blocks. The independent observation holds with orthogonal transmit waveforms and AWGN \cite{11570946}.
The position-velocity coupling matrix $\mathbf{F}_\text{PV}$ is usually neglected in related works \cite{7131233,10615966,10646234,10419729,10693955,11202824,10266619,5393291,11231051,pucci2025,xu2026}.
The EFIMs of position and velocity are obtained via the Schur complement by treating the other as the nuisance parameters:
\begin{align}\label{eq.efim}
    \mathbf{F}'_\text{P}=\mathbf{F}_\text{P}-\mathbf{F}_\text{PV}\mathbf{F}_\text{V}^{-1}\mathbf{F}_\text{PV}^T,\quad \mathbf{F}'_\text{V}=\mathbf{F}_\text{V}-\mathbf{F}_\text{PV}^T\mathbf{F}_\text{V}^{-1}\mathbf{F}_\text{PV}.
\end{align}

The CRLBs of position and velocity are denoted by the squares of PEB and VEB:
\begin{align}\label{eq.crlb}
    &C_\text{P}=\text{tr}(\mathbf{F}_\text{P}'^{-1})=\text{PEB}^2,\quad
    C_\text{V}=\text{tr}(\mathbf{F}_\text{V}'^{-1})=\text{VEB}^2.
\end{align}

In this section, the full PEB and VEB for different sensing types (MMS, MBS, and MXS) are calculated.
\vspace{-1em}

\subsection{MMS}

For monostatic radar, the delay and Doppler are given by 
\begin{subequations}\begin{align}
    &\negthinspace \tau_{m}=\frac{2r_{m}}{c_0}=\frac{2\Vert\mathbf{p}-\mathbf{p}_m\Vert}{c_0}=\frac{2\sqrt{\Delta x_m^2+\Delta y_m^2}}{c_0},\label{eq.taumono}\\
    &\negthinspace f_{m}^\text{D}\negthinspace=\negthinspace-\frac{2}{\lambda}v_{\text{r},m}\negthinspace=\negthinspace-\frac{2}{\lambda}\mathbf{v}^T\mathbf{u}_m\negthinspace=\negthinspace-\frac{2}{\lambda}(v_x\cos(\varphi_m)\negthinspace+\negthinspace v_y\sin(\varphi_m)),\negthinspace \label{eq.fmono}
\end{align}\end{subequations}
where $\mathbf{p}=[x,y]^T$ and $\mathbf{v}=[v_x,v_y]^T$ denote the target position and velocity. $v_{\text{r},m}$ represents the radial velocity with $v_{\text{r},m}>0$ denoting the motion away from the node.
$\mathbf{p}_m=[x_m,y_m]^T$ is the position of the \ac{TRx} node.
$\mathbf{u}_m=[\frac{\Delta x_m}{r_m},\frac{\Delta y_m}{r_m}]^T=[\cos(\varphi_m),\sin(\varphi_m)]^T=[c_m,s_m]^T$ represent the radial direction vector, where
$\Delta x_m=x-x_m,\ \Delta y_m=y-y_m$, and
$\varphi_m=\arctan({\Delta y_m}/{\Delta x_m})$. $c_0$ is the speed of light in vacuum. 

The Jacobian matrix of monostatic radars can be written as
 
\begin{align}{\setlength{\arraycolsep}{2.5pt}
    \mathbf{J}_m=\begin{bmatrix}
        c_m & s_m & 0 & 0\\
        a_m & b_m & c_m & s_m\\
        -\frac{s_m}{r_m} & \frac{c_m}{r_m} & 0 & 0
    \end{bmatrix},}
\end{align}
where
\begin{align}\label{eq.mmsch}
    &\negthinspace c_m=\frac{\Delta x_m}{r_m},\quad s_m=\frac{\Delta y_m}{r_m}\nonumber\\
    &\negthinspace a_m\negthinspace=\negthinspace\frac{\negthinspace v_x\negthinspace}{\negthinspace r_m\negthinspace}\negthinspace-\negthinspace\frac{\negthinspace(\Delta x_mv_x\negthinspace+\negthinspace\Delta y_mv_y)\Delta x_m\negthinspace}{r_m^3}\negthinspace=\negthinspace-\frac{\negthinspace s_m\negthinspace}{\negthinspace r_m\negthinspace}v_{\perp m},\negthinspace\nonumber \\
    &\negthinspace b_m\negthinspace=\negthinspace\frac{v_y}{r_m}\negthinspace-\negthinspace\frac{(\Delta x_mv_x\negthinspace+\negthinspace\Delta y_mv_y)\Delta y_m}{r_m^3}\negthinspace=\negthinspace\frac{c_m}{r_m}v_{\perp m},\negthinspace\nonumber\\
    &\negthinspace v_{\perp m}=\mathbf{v}^T\mathbf{u}_{\perp m},\quad \mathbf{u}_{\perp m}=[-s_m,c_m]^T.
\end{align}
$v_{\perp m}$ and $\mathbf{u}_{\perp m}$ are the velocity and unit vector in the tangential direction perpendicular to the $m$-th TRx-target radial direction. The global FIM is calculated by (\ref{eq.FIM}), where the elements are described in (\ref{eq.monoele}), in which $\mathbf{g}_m=[a_m,b_m]^T={v_{\perp m}\mathbf{u}_{\perp m}}/{r_m}$, also in the tangential subspace with an amplitude of the ratio between the tangential velocity and range. 

The EFIM is derived in (\ref{eq.efimmono}).
\begin{figure*}
{\setlength{\arraycolsep}{2pt} \begin{align}
    &\mathbf{F}_\text{P}=\sum_{m=1}^M\begin{bmatrix}
        c_m^2E_{r,m}^\text{M}+a_m^2E_{v_\text{r},m}^\text{M}+\frac{s_m^2}{r_m^2}E_{\theta,m} & c_ms_mE_{r,m}^\text{M}+a_mb_mE_{v_\text{r},m}^\text{M}-\frac{c_ms_m}{r_m^2}E_{\theta,m}\\
        c_ms_mE_{r,m}^\text{M} + a_mb_mE_{v_\text{r},m}^\text{M}-\frac{c_ms_m}{r_m^2}E_{\theta,m} & s_m^2E_{r,m}^\text{M}+b_m^2E_{v_\text{r},m}^\text{M}+\frac{c_m^2}{r_m^2}E_{\theta,m}
    \end{bmatrix}\nonumber\\
    &\quad\ =\sum_{m=1}^M\Big(E_{r,m}^\text{M} \mathbf{u}_m\mathbf{u}_m^T+E_{v_\text{r},m}^\text{M}\mathbf{g}_m\mathbf{g}_m^T+\frac{1}{r_m^2}{E_{\theta,m}}{\mathbf{u}_{\perp m}\mathbf{u}_{\perp m}^T}\Big),\nonumber
    \\
    &\mathbf{F}_\text{PV}\negthinspace=\negthinspace\sum_{m=1}^M\begin{bmatrix}
        a_mc_mE_{v_\text{r},m}^\text{M} & a_ms_mE_{v_\text{r},m}^\text{M}\\b_mc_mE_{v_\text{r},m}^\text{M} & b_ms_mE_{v_\text{r},m}^\text{M}
    \end{bmatrix}\negthinspace=\negthinspace\sum_{m=1}^ME_{v_\text{r},m}^\text{M}\mathbf{g}_m\mathbf{u}_m^T,\ \mathbf{F}_\text{V}\negthinspace=\negthinspace\sum_{m=1}^M\begin{bmatrix}
        c_m^2E_{v_\text{r},m}^\text{M} & c_ms_mE_{v_\text{r},m}^\text{M}\\
        c_ms_mE_{v_\text{r},m}^\text{M} & s_m^2E_{v_\text{r},m}^\text{M}
    \end{bmatrix}\negthinspace=\negthinspace\sum_{m=1}^ME_{v_\text{r},m}^\text{M}\mathbf{u}_m\mathbf{u}_m^T. \label{eq.monoele}
\end{align}}
\begin{subequations}\label{eq.efimmono}\begin{align}
    &\negthinspace\mathbf{F}'_\text{\negthinspace P}\negthinspace=\negthinspace\negthinspace\underbrace{\sum_{\negthinspace m=1\negthinspace}^M\negthinspace\Big(\negthinspace E_{r,m}^\text{M} \mathbf{u}_m\mathbf{u}_m^T\negthinspace+\negthinspace\frac{1}{r_m^2}E_{\theta,m} \mathbf{u}_{\perp m}\mathbf{u}_{\perp m}^T\negthinspace\Big)}_{\mathbf{F}_{\text{P}0}}\negthinspace\negthinspace+\negthinspace\overbrace{\underbrace{\sum_{\negthinspace m=1\negthinspace}^M\negthinspace E_{v_\text{r},m}^\text{M}\mathbf{g}_m\mathbf{g}_m^T}_{\mathbf{T}_\text{V}}\negthinspace-\negthinspace\Big(\negthinspace\sum_{\negthinspace m=1\negthinspace}^M\negthinspace E_{v_\text{r},m}^\text{M}\mathbf{g}_m\mathbf{u}_m^T\negthinspace\Big)\negthinspace\Big(\negthinspace\sum_{\negthinspace m=1\negthinspace}^M\negthinspace E_{v_\text{r},m}^\text{M}\mathbf{u}_m\mathbf{u}_m^T\negthinspace\Big)^{\negthinspace\negthinspace-1}\negthinspace\Big(\negthinspace\sum_{\negthinspace m=1\negthinspace}^M\negthinspace E_{v_\text{r},m}^\text{M}\mathbf{u}_m\mathbf{g}_m^T\negthinspace\Big)}^{\pmb{\Delta}_\text{V}}\negthinspace,\negthinspace\\
    &\negthinspace\mathbf{F}'_\text{V}\negthinspace=\negthinspace\underbrace{\sum_{m=1}^M\negthinspace E_{v_\text{r},m}^\text{M}\mathbf{u}_m\mathbf{u}_m^T}_{\mathbf{F}_{\text{V}0}=\mathbf{F}_\text{V}}\negthinspace-\negthinspace\underbrace{\Big(\negthinspace\sum_{m=1}^M\negthinspace  E_{v_\text{r},m}^\text{M}\mathbf{u}_m\mathbf{g}_m^T\Big)\negthinspace\Big(\negthinspace\sum_{m=1}^M\negthinspace E_{r,m}^\text{M} \mathbf{u}_m\mathbf{u}_m^T\negthinspace+\negthinspace E_{v_\text{r},m}^\text{M}\mathbf{g}_m\mathbf{g}_m^T\negthinspace+\negthinspace{\frac{1}{r_m^2}E_{\theta,m}}{\mathbf{u}_{\perp m}\mathbf{u}_{\perp m}^T}\negthinspace\Big)^{\negthinspace-1}\negthinspace\Big(\negthinspace\sum_{m=1}^M\negthinspace E_{v_\text{r},m}^\text{M}\mathbf{g}_m\mathbf{u}_m^T\negthinspace\Big)}_{\pmb{\Delta}_\text{P}}\negthinspace.\negthinspace\label{eq.efimmonov}
\end{align}\end{subequations}\end{figure*}
In $\mathbf{F}_\text{P}'$, the term of $\mathbf{F}_{\text{P}0}$ implies the information contributed by range and AoA measurements, whose contributions respectively exist in the radial and tangential subspaces.
The existence of $\pmb{\Delta}_\text{V}$ implies that the knowledge of radial velocity ($E_{v_\text{r},m}$) can contribute to position information through the sensitivity to the tangential motion (included in $\mathbf{g}_m$). 
Note that the term $\mathbf{T}_\text{V}$ comes from the position block $\mathbf{F}_\text{P}$, constituting $\pmb{\Delta}_\text{V}$ with the coupling terms together.
Consequently, the PEB depends on the target velocity, which is undesired in optimizations since the target velocity is generally unknown.

Similarly, in $\mathbf{F}_\text{V}'$, the first term ($\mathbf{F}_{\text{V}0}=\mathbf{F}_\text{V}$) denotes the information obtained by the observation of radial velocities, while the coupling term $\pmb{\Delta}_\text{P}$ introduces an information loss produced by the position uncertainty and tangential motion, reducing the available velocity information in radial subspace.
\vspace{-1em}

\subsection{MBS}
For bistatic radar, the delay and Doppler can be written as
\begin{subequations}\begin{align}
    &\tau_{m}=\frac{r_{m}}{c_0}=\frac{r_{m}^\text{t}+r_m^\text{r}}{c_0}=\frac{\Vert\mathbf{p}-\mathbf{p}_m^\text{t}\Vert+\Vert\mathbf{p}-\mathbf{p}_m^\text{r}\Vert}{c_0}\nonumber\\
    &=\frac{\sqrt{\Delta (x_m^\text{t})^{2}+(\Delta y_m^\text{t})^2}+\sqrt{\Delta (x_m^\text{r})^{2}+(\Delta y_m^\text{r})^2}}{c_0},\\
    &f_{m}^\text{D}=-\frac{v_{\text{r},m}}{\lambda}=-\frac{1}{\lambda}(v_{\text{r},m}^\text{t}+v_{\text{r},m}^\text{r})=-\frac{1}{\lambda}\mathbf{v}^T(\mathbf{u}_m^\text{t}+\mathbf{u}_m^\text{r})\nonumber\\
    &\negthinspace=\negthinspace-\frac{1}{\lambda}\negthinspace(v_x(\cos(\negthinspace\varphi_m^\text{t}\negthinspace)\negthinspace+\negthinspace\cos(\negthinspace\varphi_m^\text{r}\negthinspace))\negthinspace+\negthinspace v_y(\sin(\negthinspace\varphi_m^\text{t}\negthinspace)\negthinspace+\negthinspace\sin(\negthinspace\varphi_m^\text{r}\negthinspace))),\negthinspace
\end{align}\end{subequations}
where $\mathbf{p}_m^\text{t}=[x_m^\text{t},y_m^\text{t}]^T$ and $\mathbf{p}_m^\text{r}=[x_m^\text{r},y_m^\text{r}]^T$ represent the positions of the Tx and Rx nodes corresponding to the $m$-th sensing link.
$\mathbf{u}_m^\text{t}=[\cos(\varphi_m^\text{t}),\sin(\varphi_m^\text{t})]^T=[c_m^\text{t},s_m^\text{t}]^T$ and $\mathbf{u}_m^\text{r}=[\cos(\varphi_m^\text{r}),\sin(\varphi_m^\text{r})]^T=[c_m^\text{r},s_m^\text{r}]^T$ are the Tx and Rx direction vectors. $\varphi_m^\text{t}=\arctan({\Delta y_m^\text{t}}/{\Delta x_m^\text{t}})$ and $\varphi_m^\text{r}=\arctan({\Delta y_m^\text{r}}/{\Delta x_m^\text{r}})$ are the Tx and Rx bearing angles. For bistatic radar, we denote the total range and radial velocity of the $m$-th link as $r_m=r_m^\text{t}+r_m^\text{r}$ and $v_{\text{r},m}=v_{\text{r},m}^\text{t}+v_{\text{r},m}^\text{r}$.

The Jacobian matrix of bistatic radars can be written as
\begin{align}{\setlength{\arraycolsep}{2.5pt}
    \mathbf{J}_m=\begin{bmatrix}
        c_m & s_m & 0 & 0\\
        a_m & b_m & c_m & s_m\\
        -\frac{s_m^\text{r}}{r_m^\text{r}} & \frac{c_m^\text{r}}{r_m^\text{r}} & 0 & 0
    \end{bmatrix}\negthinspace,\negthinspace}
\end{align}
where $c_m=c_m^\text{t}+c_m^\text{r}, s_m=s_m^\text{t}+s_m^\text{r}, a_m=a_m^\text{t}+a_m^\text{r}, b_m=b_m^\text{t}+b_m^\text{r}, \mathbf{u}_m=\mathbf{u}_m^\text{t}+\mathbf{u}_m^\text{r}$, and
\begin{align}\label{eq.mbscomp}
    &c_m^\text{t}=\frac{\Delta x_m^\text{t}}{r_m^\text{t}},\ c_m^\text{r}=\frac{\Delta x_m^\text{r}}{r_m^\text{r}},\ s_m^\text{t}=\frac{\Delta y_m^\text{t}}{r_m^\text{t}},\ s_m^\text{r}=\frac{\Delta y_m^\text{r}}{r_m^\text{r}},\nonumber\\
    &a_m^\text{t}=\frac{ v_x}{ r_m^\text{t}}-\frac{(\Delta x_m^\text{t}v_x+\Delta y_m^\text{t}v_y)\Delta x_m^\text{t}}{(r_m^\text{t})^3}=-\frac{ s_m^\text{t}}{ r_m^\text{t}}v_{\perp m}^\text{t},\nonumber\\ 
    &a_m^\text{r}=\frac{ v_x}{ r_m^\text{r}}-\frac{(\Delta x_m^\text{r}v_x+\Delta y_m^\text{r}v_y)\Delta x_m^\text{r}}{(r_m^\text{r})^3}=-\frac{ s_m^\text{r}}{ r_m^\text{r}}v_{\perp m}^\text{r},\nonumber\\
    &b_m^\text{t}=\frac{ v_y}{ r_m^\text{t}}-\frac{(\Delta x_m^\text{t}v_x+\Delta y_m^\text{t}v_y)\Delta y_m^\text{t}}{(r_m^\text{t})^3}=\frac{ c_m^\text{t}}{ r_m^\text{t}}v_{\perp m}^\text{t},\nonumber\\ 
    &b_m^\text{r}=\frac{ v_y}{ r_m^\text{r}}-\frac{(\Delta x_m^\text{r}v_x+\Delta y_m^\text{r}v_y)\Delta y_m^\text{r}}{(r_m^\text{r})^3}=\frac{ c_m^\text{r}}{ r_m^\text{r}}v_{\perp m}^\text{r},\nonumber\\
    &v_{\perp m}^\text{t}=\mathbf{v}^T\mathbf{u}_{\perp m}^\text{t},\ v_{\perp m}^\text{r}=\mathbf{v}^T\mathbf{u}_{\perp m}^\text{r},\nonumber\\
    &\mathbf{u}_{\perp m}^\text{t}=[-s_m^\text{t},c_m^\text{t}]^T,\ \mathbf{u}_{\perp m}^\text{r}=[-s_m^\text{r},c_m^\text{r}]^T,\nonumber\\ &v_{\perp m}=v_{\perp m}^\text{t}+v_{\perp m}^\text{r},\ \mathbf{u}_{\perp m}=\mathbf{u}_{\perp m}^\text{t}+\mathbf{u}_{\perp m}^\text{r}.
\end{align}

The global FIM follows the calculation in (\ref{eq.FIM}), whereas the elements are given in (\ref{eq.biele}), in which $\mathbf{g}_m=[a_m,b_m]^T={v_{\perp m}^\text{t}\mathbf{u}_{\perp m}^\text{t}}/{r_m^\text{t}}+{v_{\perp m}^\text{r}\mathbf{u}_{\perp m}^\text{r}}/{r_m^\text{r}}$. The elements have a similar definition as MMS, only with changes in the signal-level terms $E_{r,m}^\text{B},E_{v_\text{r},m}^\text{B}$ explained in Section~\ref{sec.model} and the components defined regarding the bistatic radar characteristics in (\ref{eq.mbscomp}). Note that the AoA-related terms only include the components of Rx, hence $C_\theta$ has the same formula for MMS and MBS, and its corresponding term ${E_{\theta,m}}{\mathbf{u}_{\perp m}^\text{r}(\mathbf{u}_{\perp m}^\text{r}})^T/(r_m^\text{r})^2$ only involves the Rx-related parameters.

\begin{figure*}
{\setlength{\arraycolsep}{2pt} \begin{align}
    &\mathbf{F}_\text{P}=\sum_{m=1}^M\begin{bmatrix}
        c_m^2E_{r,m}^\text{B}+a_m^2E_{v_\text{r},m}^\text{B}+\frac{(s_m^\text{r})^2}{(r_m^\text{r})^2}E_{\theta,m} & c_ms_mE_{r,m}^\text{B}+a_mb_mE_{v_\text{r},m}^\text{B}-\frac{c_m^\text{r}s_m^\text{r}}{(r_m^\text{r})^2}E_{\theta,m}\\
        c_ms_mE_{r,m}^\text{B} + a_mb_mE_{v_\text{r},m}^\text{B}-\frac{c_m^\text{r}s_m^\text{r}}{(r_m^\text{r})^2}E_{\theta,m} & s_m^2E_{r,m}^\text{B}+b_m^2E_{v_\text{r},m}^\text{B}+\frac{(c_m^\text{r})^2}{(r_m^\text{r})^2}E_{\theta,m}
    \end{bmatrix}\nonumber\\
    &\quad\ =\sum_{m=1}^M\Big(E_{r,m}^\text{B} \mathbf{u}_m\mathbf{u}_m^T+E_{v_\text{r},m}^\text{B}\mathbf{g}_m\mathbf{g}_m^T+\frac{1}{(r_m^\text{r})^2}{E_{\theta,m}}{\mathbf{u}_{\perp m}^\text{r}(\mathbf{u}_{\perp m}^\text{r}})^T\Big),\nonumber
    \\
    &\mathbf{F}_\text{PV}\negthinspace=\negthinspace\sum_{m=1}^M\begin{bmatrix}
        a_mc_mE_{v_\text{r},m}^\text{B} & a_ms_mE_{v_\text{r},m}^\text{B}\\b_mc_mE_{v_\text{r},m}^\text{B} & b_ms_mE_{v_\text{r},m}^\text{B}
    \end{bmatrix}\negthinspace=\negthinspace\sum_{m=1}^ME_{v_\text{r},m}^\text{B}\mathbf{g}_m\mathbf{u}_m^T,\ \mathbf{F}_\text{V}\negthinspace=\negthinspace\sum_{m=1}^M\begin{bmatrix}
        c_m^2E_{v_\text{r},m}^\text{B} & c_ms_mE_{v_\text{r},m}^\text{B}\\
        c_ms_mE_{v_\text{r},m}^\text{B} & s_m^2E_{v_\text{r},m}^\text{B}
    \end{bmatrix}\negthinspace=\negthinspace\sum_{m=1}^ME_{v_\text{r},m}^\text{B}\mathbf{u}_m\mathbf{u}_m^T. \label{eq.biele}
\end{align}}
\begin{subequations}\label{eq.efimbi}\begin{align}
    &\negthinspace\mathbf{F}'_\text{\negthinspace P}\negthinspace=\negthinspace\negthinspace\underbrace{\sum_{\negthinspace m=1\negthinspace}^M\negthinspace\negthinspace\Big(\negthinspace E_{r,m}^\text{B}\negthinspace \mathbf{u}_m\mathbf{u}_m^T\negthinspace+\negthinspace\frac{1}{\negthinspace(r_{\negthinspace m}^\text{r})^2\negthinspace}E_{\negthinspace\theta,m}\negthinspace \mathbf{u}_{\perp m}^\text{r}\negthinspace(\negthinspace\mathbf{u}_{\perp m}^\text{r}\negthinspace)\negthinspace^T\negthinspace\Big)}_{\mathbf{F}_{\text{P}0}}\negthinspace\negthinspace+\negthinspace\overbrace{\underbrace{\sum_{\negthinspace m=1\negthinspace}^M\negthinspace E_{v_\text{r},m}^\text{B}\mathbf{g}_m\mathbf{g}_m^T}_{\mathbf{T}_\text{V}}\negthinspace-\negthinspace\Big(\negthinspace\sum_{\negthinspace m=1\negthinspace}^M\negthinspace E_{v_\text{r},m}^\text{B}\mathbf{g}_m\negthinspace\mathbf{u}_m^T\negthinspace\Big)\negthinspace\Big(\negthinspace\sum_{\negthinspace m=1\negthinspace}^M\negthinspace E_{v_\text{r},m}^\text{B}\negthinspace\mathbf{u}_m\mathbf{u}_m^T\negthinspace\Big)^{\negthinspace\negthinspace-1}\negthinspace\negthinspace\Big(\negthinspace\sum_{\negthinspace m=1\negthinspace}^M\negthinspace E_{v_\text{r},m}^\text{B}\negthinspace\mathbf{u}_m\mathbf{g}_m^T\negthinspace\Big)}^{\pmb{\Delta}_\text{V}}\negthinspace,\negthinspace\\
    &\negthinspace\mathbf{F}'_\text{V}\negthinspace=\negthinspace\negthinspace\underbrace{\sum_{\negthinspace m=1\negthinspace}^M\negthinspace E_{v_\text{r},m}^\text{B}\negthinspace\mathbf{u}_m\mathbf{u}_m^T}_{\mathbf{F}_{\text{V}0}=\mathbf{F}_\text{V}}\negthinspace-\negthinspace\underbrace{\Big(\negthinspace\sum_{\negthinspace m=1\negthinspace}^M\negthinspace  E_{v_\text{r},m}^\text{B}\negthinspace\mathbf{u}_m\mathbf{g}_m^T\Big)\negthinspace\Big(\negthinspace\sum_{\negthinspace m=1\negthinspace}^M\negthinspace E_{r,m}^\text{B}\negthinspace \mathbf{u}_m\mathbf{u}_m^T\negthinspace+\negthinspace E_{v_\text{r},m}^\text{B}\mathbf{g}_m\mathbf{g}_m^T\negthinspace+\negthinspace{\frac{1}{\negthinspace(r_m^\text{r})^2\negthinspace}E_{\theta,m}}\negthinspace{\mathbf{u}_{\perp m}^\text{r}(\negthinspace\mathbf{u}_{\perp m}^\text{r}\negthinspace)^T}\negthinspace\Big)^{\negthinspace-1}\negthinspace\negthinspace\Big(\negthinspace\sum_{\negthinspace m=1\negthinspace}^M\negthinspace E_{v_\text{r},m}^\text{B}\mathbf{g}_m\mathbf{u}_m^T\negthinspace\Big)}_{\pmb{\Delta}_\text{P}}\negthinspace.\negthinspace\label{eq.efimbiv}
\end{align}\end{subequations}\end{figure*}

The position and velocity EFIMs are given in (\ref{eq.efimbi}). For MBS, the range and radial velocity measurements contribute along the bistatic radial directions determined by both the Tx and Rx nodes, thereby providing richer spatial diversity than MMS. $\pmb{\Delta}_\text{V}$ and $\pmb{\Delta}_\text{P}$ denote the coupling effects that mainly influence the tangential position information and radial velocity information, respectively, which effects are related to the tangential velocity and uncertainty in position and radial velocity.

Note that the FIMs and CRLBs of MMS and MBS have a consistent expression, where MMS can be regarded as a special case of MBS with collocated Tx and Rx nodes for each link. In this case, the resulting FIMs of MMS and MBS are identical. This implies that the CRLBs of MMS and MBS mainly differs in spatial diversity instead of the amount of acquirable information for each link.

\subsection{MXS}
MXS can be regarded as a hybrid of MMS and MBS, which is also named heterogeneous sensing in \cite{pucci2025}. The nodes can operate as monostatic TRx, bistatic Tx or Rx, or a hybrid of both. 
The resulting global FIM is the sum of the individual monostatic and bistatic links, i.e.,
\begin{align}
    \mathbf{F}=\sum_{m\in\mathcal{M}}\mathbf{F}_m+\sum_{n\in\mathcal{B}}\mathbf{F}_n,
\end{align}
where $\mathcal{M}$ and $\mathcal{B}$ denote the sets of monostatic and bistatic links, respectively, and $|\mathcal{M}|+|\mathcal{B}|=M$. The elements $\mathbf{F}_\text{P},\mathbf{F}_\text{V},\text{ and }\mathbf{F}_\text{PV}$ are the sum of those of the individual monostatic and bistatic links; the EFIMs and CRLBs are calculated following (\ref{eq.efim}) and (\ref{eq.crlb}). 

A problem is that the formulas of PEB and VEB are highly complicated and intractable. Consequently, the application of optimization tools, like \ac{SDP} and \ac{SCA}, becomes difficult due to the existence of inverse matrix calculation in $\pmb{\Delta}_\text{P}$ and $\pmb{\Delta}_\text{V}$ and the unknown velocity. For optimization problems, a tractable target function maintaining high accuracy with respect to the full CRLB is desired.

\section{Tractable CRLB and Validity Conditions}\label{sec.simp}
The tractable EFIM requires removing the coupling terms, where the inverse matrix calculation leads to optimization-unfriendly formulations, and eliminating the impact of unknown velocity carried in $\mathbf{g}_m$, since the existence of an unknown state parameter is not expected in the target function. Therefore, the strategy is to simplify the position and velocity EFIMs to $\mathbf{F}_{\text{P}0}$ and $\mathbf{F}_{\text{V}0}=\mathbf{F}_\text{V}$. 
Although the corresponding $C_\text{P}$ and $C_\text{V}$ are still non-convex, convexation approaches like \ac{SDP} can be easily applied. 
To this end, it is necessary to evaluate the accuracy loss due to the simplification and establish the validity conditions under which the simplification can be safely applied.

Before analyzing the validity conditions for individual sensing types, we first introduce a general guiding principle:
\begin{Proposition}\label{po.1}
    If $\epsilon \mathbf{F}_{\mathrm{P}0}-\mathbf{T}_\mathrm{V}\succeq0,\ 0<\epsilon\ll1$, the PEB and VEB derived from the simplified EFIM closely approximate those derived from the full EFIM.
\end{Proposition}
\begin{proof}
    Starting from position EFIM, $\pmb{\Delta}_\text{V}=\mathbf{T}_\text{V}-\mathbf{F}_\text{PV}\mathbf{F}_\text{V}^{-1}\mathbf{F}_\text{PV}^T$ is \ac{PSD} since the block matrix is PSD:
    \begin{align}\label{eq.schur}
        \begin{bmatrix}
            \mathbf{T}_\text{V} & \mathbf{F}_\text{PV}\\
            \mathbf{F}_\text{PV}^T & \mathbf{F}_\text{V}
        \end{bmatrix}=\sum_{m=1}^M E_{v_\text{r},m}\begin{bmatrix}
            \mathbf{g}_m \\ \mathbf{u}_m
        \end{bmatrix}\begin{bmatrix}
            \mathbf{g}_m \\ \mathbf{u}_m
        \end{bmatrix}^T\succeq0.
    \end{align}
    Because $\mathbf{F}_\text{V}\succ0$, the Schur complement yields $\pmb{\Delta}_\text{V}=\mathbf{T}_\text{V}-\mathbf{F}_\text{PV}\mathbf{F}_\text{V}^{-1}\mathbf{F}_\text{PV}^T\succeq0$ and $\pmb{\Delta}_\text{V}\preceq\mathbf{T}_\text{V}$.

    If $\epsilon \mathbf{F}_{\text{P}0}-\mathbf{T}_\text{V}\succeq0$, we have
    \begin{align}
        (1+\epsilon)\mathbf{F}_{\text{P}0}\succeq \mathbf{F}_{\text{P}0}+\mathbf{T}_\text{V}\succeq \mathbf{F}'_\text{P}\succeq\mathbf{F}_{\text{P}0}.
    \end{align}
    According to the property of PSD matrices,
    \begin{align}
        \text{tr}(\mathbf{F}_{\text{P}0}^{-1})\negthinspace\ge\negthinspace \text{tr}((\mathbf{F}_{\text{P}}')^{-1})\negthinspace\ge\negthinspace \frac{\text{tr}(\mathbf{F}_{\text{P}0}^{-1})}{1+\epsilon}\Leftrightarrow C_{\text{P}0}\negthinspace\ge\negthinspace C_\text{P}\negthinspace\ge\negthinspace\frac{C_{\text{P}0}}{1+\epsilon},\negthinspace
    \end{align}
    where $C_{\text{P}0}$ denotes $\text{PEB}^2$ calculated based on $\mathbf{F}_{\text{P}0}$. If $0<\epsilon\ll1$, $C_{\text{P}0}\approx C_{\text{P}}$, i.e., the PEBs derived from simplified and full EFIMs are closely approximated.

    Furthermore, since $\mathbf{T}_\text{V}\succeq0$, the Schur complement from (\ref{eq.schur}) yields $\mathbf{F}_\text{V}-\mathbf{F}_\text{PV}^T\mathbf{T}_\text{V}^{\dagger}\mathbf{F}_\text{PV}\succeq0$. If $\epsilon \mathbf{F}_{\text{P}0}-\mathbf{T}_\text{V}\succeq0$, we have $\pmb{\Delta}_\text{P}=\mathbf{F}_\text{PV}^T(\mathbf{F}_{\text{P}0}+\mathbf{T}_\text{V})^{-1}\mathbf{F}_\text{PV}\preceq\frac{\epsilon}{1+\epsilon}\mathbf{F}_\text{PV}^T\mathbf{T}_\text{V}^{\dagger}\mathbf{F}_\text{PV}$, thus
    \begin{align}
        \mathbf{F}_\text{V}\succeq \mathbf{F}_\text{V}'\succeq\mathbf{F}_\text{V}-\frac{\epsilon}{1+\epsilon}\mathbf{F}_\text{PV}^T\mathbf{T}_\text{V}^{\dagger}\mathbf{F}_\text{PV}\succeq \frac{1}{1+\epsilon}\mathbf{F}_\text{V},
    \end{align}
    \begin{align}
        \negthinspace\negthinspace\text{tr}(\mathbf{F}_\text{V})\negthinspace\le\negthinspace \text{tr}(\mathbf{F}_\text{V}')\negthinspace\le\negthinspace(1\negthinspace+\negthinspace\epsilon)\text{tr}(\mathbf{F}_\text{V})\negthinspace\Leftrightarrow\negthinspace C_{\text{V}0}\negthinspace\le\negthinspace C_\text{V}\negthinspace\le\negthinspace(1\negthinspace+\negthinspace\epsilon)C_{\text{V}0},\negthinspace\negthinspace
    \end{align}
    where $C_{\text{V}0}$ denotes $\text{VEB}^2$ calculated based on $\mathbf{F}_{\text{V}0}=\mathbf{F}_{\text{V}}$. If $0<\epsilon\ll1$, $C_{\text{V}0}\approx C_{\text{V}}$.
\end{proof}

\textbf{Proposition~\ref{po.1}} provides a sufficient condition for safely adopting the simplified CRLB. However, this condition is purely mathematical and lacks physical interpretation. Therefore, a more geometrical condition is derived based on \textbf{Proposition~\ref{po.1}} in this section, taking the characteristics of different sensing types into account.

\subsection{MMS}\label{sec.simpmms}
For MMS, the simplified position and velocity EFIMs can be written as
\begin{subequations}\label{eq.mmssimp}\begin{align}
    &\mathbf{F}_{\text{P}0}=\sum_{ m=1}^M\Big( E_{r,m}^\text{M} \mathbf{u}_m\mathbf{u}_m^T+\frac{1}{r_m^2}E_{\theta,m} \mathbf{u}_{\perp m}\mathbf{u}_{\perp m}^T\Big),\\
    &\mathbf{F}_{\text{V}0}=\sum_{m=1}^M E_{v_\text{r},m}^\text{M}\mathbf{u}_m\mathbf{u}_m^T.
\end{align}\end{subequations}

The condition $\epsilon \mathbf{F}_{\mathrm{P}0}-\mathbf{T}_\mathrm{V}\succeq0$ is equivalent to
\begin{align}\label{eq.mmsprop}
    &\sum_{\negthinspace m=1\negthinspace}^M\negthinspace\Big[\epsilon E_{r,m}^\text{M}(\mathbf{u}_m^T\mathbf{x})^2\negthinspace+\negthinspace\frac{1}{r_m^2}\Big(\negthinspace\epsilon{E_{\theta,m}}-{v_{\perp m}^2E_{v_\text{r},m}^\text{M}}\negthinspace\Big)(\mathbf{u}_{\perp m}^T\mathbf{x})^2\Big]\negthinspace\ge\negthinspace 0,\negthinspace\nonumber\\
    &\forall\, ||\mathbf{x}||^2=1.
\end{align} 
Since for each link, the terms of $E_{\theta,m}$ and $E_{v_\text{r},m}^\text{M}$ lie in the same direction $\mathbf{u}_{\perp m}$, if $E_{\theta,m}\gg v_{\perp m}^2E_{v_\text{r},m}^\text{M},\forall m\in \mathcal{M}$, the criterion of \textbf{Proposition~\ref{po.1}} safely holds. Considering the signal-level EFIM components (\ref{eq.sige}), $E_{\theta,m}\gg v_{\perp m}^2E_{v_\text{r},m}^\text{M}$ corresponds to
\begin{align}\label{eq.tan}
    &\cos^2(\theta)(N_\text{r}^2-1)\gg 4\Bar{v}_{\perp m}^2,
\end{align}
where $\Bar{v}_{\perp m}$ denotes the tangential velocity normalized to the velocity resolution $\Delta v=\frac{\lambda}{2T_\text{F}}$. However, the condition in (\ref{eq.tan}) is typically restrictive under practical system configurations. For example, for a 5G FR2 signal with carrier frequency $f_\text{c}=28$\,GHz, $T_\text{F}=10$\,ms, the velocity resolution is $\Delta v\approx0.54$\,m/s. Assuming the target vehicle lies in the normal direction of the Rx antenna ($\theta=0$) with a tangential velocity of $20$\,m/s, (\ref{eq.tan}) requires a horizontal antenna size of $N_\text{r}\gg75$, which is generally unsatisfied in realistic scenarios. Besides, the direction term $\cos^2(\theta)$ makes the issue more critical.

Another solution is leveraging $E_{r,m}^\text{M}$.
Rewrite (\ref{eq.mmsprop}) by only maintaining the terms of $E_{r,m}^\text{M}$ and $E_{v_\text{r},m}^\text{M}$, yielding
\begin{align}\label{eq.rv}
    \sum_{ m=1}^M\Big[\epsilon E_{r,m}^\text{M}(\mathbf{u}_m^T\mathbf{x})^2-\frac{v_{\perp m}^2}{r_m^2}{E_{v_\text{r},m}^\text{M}}(\mathbf{u}_{\perp m}^T\mathbf{x})^2\Big]\ge 0.
\end{align}
For each link, the contribution of range and radial velocity information exists in orthogonal subspaces. To make (\ref{eq.rv}) tractable, we assume that the nodes provide sufficiently diverse sensing directions, which is approximately satisfied in sensing networks with dense or well-distributed nodes, and the target should be surrounded by the cooperating nodes, e.g., lies inside the polygon formed by the nodes with sufficient distance from the polygon edges. In these scenarios, the directional diversity is guaranteed, while the fluctuation of $E_{r,m}$ and $E_{v_\text{r},m}/r_m^2$ over multiple nodes is limited or can be approximately averaged out.
In consequence, the weighted directional matrices $\sum_{ m=1}^M E_{r,m}^\text{M}\mathbf{u}_m\mathbf{u}_m^T$ and $\sum_{ m=1}^M\frac{v_{\perp m}^2}{r_m^2}{E_{v_\text{r},m}^\text{M}}\mathbf{u}_{\perp m}\mathbf{u}_{\perp m}^T$ are approximately isotropic, the directional projections $(\mathbf{u}_{ m}^T\mathbf{x})^2$ and $(\mathbf{u}_{\perp m}^T\mathbf{x})^2$ can be approximated by their angular averages, 
yielding 
\begin{subequations}\label{eq.monomat}
\begin{align}
    &\sum_{ m=1}^M E_{r,m}^\text{M}(\mathbf{u}_m^T\mathbf{x})^2\approx\frac{1}{2}\sum_{ m=1}^M E_{r,m}^\text{M},\\
    &\sum_{ m=1}^M\frac{v_{\perp m}^2}{r_m^2}{E_{v_\text{r},m}^\text{M}}(\mathbf{u}_{\perp m}^T\mathbf{x})^2\le\sum_{ m=1}^M\frac{||\mathbf{v}||^2}{r_m^2}{E_{v_\text{r},m}^\text{M}}(\mathbf{u}_{\perp m}^T\mathbf{x})^2\nonumber\\
    &\approx\frac{1}{2}\sum_{ m=1}^M\frac{||\mathbf{v}||^2}{r_m^2}{E_{v_\text{r},m}^\text{M}},
\end{align}    
\end{subequations}
where the replacement of $v_\perp$ by $\mathbf{v}$ avoids the dependence on velocity direction in the following derivation.
Consequently, (\ref{eq.rv}) can be approximately rewritten as
\begin{align}\label{eq.crinodes}
    &\epsilon\sum_{ m=1}^M E_{r,m}^\text{M}\ge\sum_{ m=1}^M\frac{||\mathbf{v||}^2}{r_m^2}{E_{v_\text{r},m}^\text{M}}\Leftrightarrow\epsilon\sum_{ m=1}^M{\gamma_m}\ge{\Bar{v}}^2\sum_{ m=1}^M\frac{\gamma_m}{\Bar{r}_m^2}\nonumber\\
    &\overset{\gamma\propto \varsigma/r^4}{\Rightarrow}\epsilon\sum_{ m=1}^M\frac{\varsigma_m}{r_m^4}\ge \Bar{v}^2\sum_{ m=1}^M\frac{\varsigma_m}{\Bar{r}_m^2r_m^4},
\end{align}
where the bandwidth, time, and antenna resources of each node are assumed to be identical. $\gamma\propto \varsigma/r^4$ comes from the monostatic radar propagation loss equation, where $\varsigma$ denotes the \ac{RCS}. $\Bar{r}_m=r_m/\Delta r$ and $\Bar{v}=||\mathbf{v}||/\Delta v$ denote the normalized range and absolute velocity with respect to their resolution. 
(\ref{eq.crinodes}) provides an interpretable and geometrical engineering criterion, being described by the SNR-weighted ratio between the normalized velocity and range.

A clearer interpretation can be obtained by the per-node analysis in (\ref{eq.crinodes}), resulting in
\begin{align}\label{eq.crinode}
    \epsilon{\Bar{r}_m^2}\ge\Bar{v}^2,\quad 0<\epsilon\ll1,\quad\forall m,
\end{align}
being transformed into a purely geometric condition. This criterion leverages the relationship between the normalized range and velocity, implying that a higher normalized range and a lower normalized velocity can reduce the difference between the full and simplified CRLBs. (\ref{eq.crinode}) does not require identical resource allocation.
In the case of sufficient spatial diversity, (\ref{eq.crinode}) can be approximately regarded as a sufficient condition under which the simplified CRLB is valid.

\subsection{MBS}
For MBS, the simplified position and velocity EFIMs can be written as
\begin{subequations}\label{eq.mbsfimp}\begin{align}
    &\mathbf{F}_{\text{P}0}\negthinspace=\negthinspace\sum_{ m=1}^M\Big( E_{r,m}^\text{B} \mathbf{u}_m\mathbf{u}_m^T+\frac{1}{(r_m^\text{r})^2}E_{\theta,m} \mathbf{u}_{\perp m}^\text{r}(\mathbf{u}_{\perp m}^\text{r})^T\Big),\negthinspace\\
    &\mathbf{F}_{\text{V}0}=\sum_{m=1}^M E_{v_\text{r},m}^\text{B}\mathbf{u}_m\mathbf{u}_m^T.
\end{align}\end{subequations}

Since the criterion based on $E_{\theta}$ is invalid, 
we rewrite $\epsilon \mathbf{F}_{\mathrm{P}0}-\mathbf{T}_\mathrm{V}\succeq0$ by only maintaining the terms of $E_{r,m}^\text{B}$ and $E_{v_\text{r},m}^\text{B}$, yielding
\begin{align}\label{eq.rvb}
    &\epsilon\sum_{ m=1}^M E_{r,m}^\text{B}\big((\mathbf{u}_m^\text{t})^T\mathbf{x}+(\mathbf{u}_m^\text{r})^T\mathbf{x}\big)^2 \ge\nonumber\\
    &\sum_{ m=1}^M{E_{v_\text{r},m}^\text{B}}\Big(\frac{ v_{\perp m}^\text{t}}{r_m^\text{t}}(\mathbf{u}_{\perp m}^\text{t})^T\mathbf{x}+\frac{ v_{\perp m}^\text{r}}{r_m^\text{r}}(\mathbf{u}_{\perp m}^\text{r})^T\mathbf{x}\Big)^2.\negthinspace
\end{align}

Under the assumption of diverse sensing direction and approximately isotropic characteristics, the cross terms $\mathbf{u}^\text{t}_m(\mathbf{u}^\text{r}_m)^T$ and $\mathbf{u}^\text{t}_{\perp m}(\mathbf{u}^\text{r}_{\perp m})^T$ are averaged out over multiple sensing links. 
The terms in (\ref{eq.rvb}) are rewritten as
\begin{subequations}\begin{align}
    &\sum_{ m=1}^M E_{r,m}^\text{B}\big((\mathbf{u}_m^\text{t})^T\mathbf{x}+(\mathbf{u}_m^\text{r})^T\mathbf{x}\big)^2\approx\sum_{ m=1}^M E_{r,m}^\text{B},\\
    &\sum_{ m=1}^M{E_{v_\text{r},m}^\text{B}}\Big(\frac{ v_{\perp m}^\text{t}}{r_m^\text{t}}(\mathbf{u}_{\perp m}^\text{t})^T\mathbf{x}+\frac{ v_{\perp m}^\text{r}}{r_m^\text{r}}(\mathbf{u}_{\perp m}^\text{r})^T\mathbf{x}\Big)^2\nonumber\\
    &\le \sum_{ m=1}^M{E_{v_\text{r},m}^\text{B}}\Big(\frac{ ||\mathbf{v}||}{r_m^\text{t}}(\mathbf{u}_{\perp m}^\text{t})^T\mathbf{x}+\frac{ ||\mathbf{v}||}{r_m^\text{r}}(\mathbf{u}_{\perp m}^\text{r})^T\mathbf{x}\Big)^2\nonumber\\
    &\approx\frac{1}{2}\sum_{ m=1}^M{E_{v_\text{r},m}^\text{B}} \Big(\frac{||\mathbf{v}||^2}{(r_m^\text{t})^2}+\frac{||\mathbf{v}||^2}{(r_m^\text{r})^2}\Big).
\end{align}    
\end{subequations}

Hence, (\ref{eq.rvb}) can be approximately rewritten as
\begin{align}\label{eq.crinodesb}
    &\epsilon\sum_{ m=1}^M E_{r,m}^\text{B}\ge \frac{1}{2}\sum_{ m=1}^M{E_{v_\text{r},m}^\text{B}} \Big(\frac{||\mathbf{v}||^2}{(r_m^\text{t})^2}+\frac{||\mathbf{v}||^2}{(r_m^\text{r})^2}\Big)\Leftrightarrow\nonumber\\
    & \epsilon\sum_{ m=1}^M\negthinspace \gamma_m\negthinspace\ge\negthinspace\frac{1}{2}\Bar{v}^2\negthinspace\sum_{ m=1}^M\negthinspace\gamma_m\Big(\frac{1}{(\Bar{r}_m^\text{t})^2}+\frac{1}{(\Bar{r}_m^\text{r})^2}\Big)\negthinspace\ge\negthinspace\Bar{v}^2\negthinspace\sum_{ m=1}^M\negthinspace\gamma_m\frac{1}{\Bar{r}_m^\text{t}\Bar{r}_m^\text{r}}\nonumber\\
    &\negthinspace\overset{\gamma\propto \varsigma/({r}_m^\text{t}{r}_m^\text{r})^2}{\Rightarrow}\negthinspace\epsilon\sum_{ m=1}^M\frac{\varsigma_m}{({r}_m^\text{t})^2({r}_m^\text{r})^2}\ge \Bar{v}^2\sum_{ m=1}^M\frac{\varsigma_m}{\Bar{r}_m^\text{t}\Bar{r}_m^\text{r}({r}_m^\text{t})^2({r}_m^\text{r})^2},
\end{align}
where the bandwidth, time, and antenna resources of each node are assumed to be identical. The clearer interpretation by per-node analysis without requirements for identical resources is given by
\begin{align}\label{eq.crinodeb}
    \epsilon{\Bar{r}_m^\text{t}\Bar{r}_m^\text{r}}\ge\Bar{v}^2, \quad 0<\epsilon\ll1,\quad\forall m.
\end{align}

Equation (\ref{eq.crinodeb}) puts forward a requirement on both the Tx and Rx ranges, while MMS can be considered as a special case of MBS when the Tx and Rx are collocated.

\begin{table*}
    \centering
    \caption{Representative 5G scenarios and parameters \cite{3gpp38913}. The form of $a_1/a_2$ implies the value of MMS and MBS, respectively.}
    \begin{tabular}{c|c|c|c|c|c|c|c|c}
    \hline
        Scenario & $f_\text{c}$\,(GHz) around & B\,(GHz) & ISD\,(m) & $v$\,(km/h) & $\Delta r$\,(m) & $\Delta v$\,(m/s) & $\Bar{v}$ & $r \text{ or } \sqrt{r^\text{t}r^\text{r}}$\,(m) \\ \hline
        \multirow{3}{*}{Indoor hotspot} & 4 & $\le0.2$ & \multirow{3}{*}{20} & \multirow{3}{*}{3} & $\ge0.75/1.5$ & 3.75/7.5 & 0.22/0.11 & $\gg0.17$ \\ \cline{2-3}\cline{6-9} & 30 & $\le1$ & & & $\ge0.15/0.3$ & 0.5/1 & 1.67/0.83 & $\gg0.25$ \\ \cline{2-3}\cline{6-9} & 70 & $\le1$ & & & $\ge0.15/0.3$ & 0.21/0.43 & 3.89/1.94 & $\gg0.58$ \\ \hline
       \multirow{2}{*}{Dense urban} & {4} & {$\le0.2$} & \multirow{2}{*}{\makecell[c]{Macro: 200\\Micro: e.g., 100}} & \multirow{2}{*}{30} & {$\ge0.75/1.5$} & {3.75/7.5} & {2.22/1.11} & {$\gg1.67$} \\ \cline{2-3}\cline{6-9}
        & 30 & $\le 1$ & & & $\ge0.15/0.3$ & 0.5/1 & 16.67/8.33 & $\gg2.5$ \\ \hline
        \multirow{2}{*}{Rural} & 0.7 & $\le0.02$ & \multirow{2}{*}{1732 or 5000} & \multirow{2}{*}{120} & $\ge7.5/15$ & 21.43/42.86 & 1.56/0.78 & $\gg11.67$ \\ \cline{2-3}\cline{6-9}
        & 4 & $\le0.2$ & & & $\ge0.75/1.5$ & 3.75/7.5 & 8.89/4.44 & $\gg6.67$ \\ \hline
        {Highway} & {6} & {$\le0.2$} & \makecell[c]{BS: 1732 or 500\\RSU: 50 or 100} & {300} & $\ge0.75/1.5$ & 2.5/5 & 33.33/16.67 & $\gg25$ \\\hline
    \end{tabular}
    \label{tab:scenario}
    \vspace{-1em}
\end{table*}

\subsection{MXS}
Since MXS is the mix of MMS and MBS, the simplified EFIMs are defined by (\ref{eq.mmssimp}) for monostatic links and (\ref{eq.mbsfimp}) for bistatic links, i.e.,
\begin{subequations}\begin{align}
    &\mathbf{F}_{\text{P}0}=\sum_{m\in\mathcal{M}}\mathbf{F}_{\text{P}0,m}+\sum_{n\in\mathcal{B}}\mathbf{F}_{\text{P}0,n},\\ &\mathbf{F}_{\text{V}0}=\sum_{m\in\mathcal{M}}\mathbf{F}_{\text{V}0,m}+\sum_{n\in\mathcal{B}}\mathbf{F}_{\text{V}0,n},
\end{align}\end{subequations}
and the criterion follows (\ref{eq.crinode}) and (\ref{eq.crinodeb}) for monostatic and bistatic links, respectively. 
The proposed criterion is summarized as \textbf{Theorem~\ref{t.1}}:

\begin{Theorem}\label{t.1}
    In a cooperative sensing network that provides sufficient spatial diversity and approximately isotropic range and velocity information over the detection area, a sufficient but conservative condition under which the simplified PEB and VEB closely approximate their full counterparts is that (\ref{eq.crinode}) holds for each monostatic link and (\ref{eq.crinodeb}) holds for each bistatic link.
\end{Theorem}

Table~\ref{tab:scenario} summarizes representative 5G traffic scenarios and corresponding parameters and requirements studied in \cite{3gpp38913}. Particularly, the requirement to $r_m$ (\ref{eq.crinode}) or $r_m^\text{t}r_m^\text{r}$ (\ref{eq.crinodeb}) are calculated based on $\Delta r$, $\Delta v$, and general speed $v$, and compared to \ac{ISD}. The 5G frame length $T_\text{F}$ is always 10\,ms.
It is clear that the presented criteria can be generally satisfied since the required range is much shorter than the ISD, especially for urban, rural, and highway scenarios, where the distance between the targets (vehicles, pedestrians, etc.) and the sensing nodes, such as BSs and \acp{RSU}, is generally lower bounded by safety spacing, lane structure, and infrastructure deployment constraints. Therefore, the criterion can be generally satisfied, implying that the simplified CRLB can be safely adopted in general cases.

{\textit{Remark:} The results in Sections~\ref{sec.full} and \ref{sec.simp} are also applicable to multi-target scenarios. Since the transformation matrix $\mathbf{J}_m$ of each target only involves its own parameters, the multi-target coupling influences the PEB and VEB only through the impact on the signal-level EFIMs $\mathbf{E}$ (\ref{eq.sige}). For instance, in a sensing link, the coupling with target B reduces the delay information of target A to $E_\text{r,A}-E_\text{r,AB}E_\text{r,B}^{-1}E_\text{r,AB}$, where $E_\text{r,AB}$ denotes the coupling term \cite{8356190}. However, it is proved in \cite{11481148} that a sufficient condition for the coupling to be negligible is that the delay or Doppler difference between the targets is larger than the corresponding resolution cell size, which is generally satisfied, especially for mmWave due to the large bandwidth and high carrier frequency. 
For instance, for a 5G waveform at $28$\,GHz with a bandwidth of 400\,MHz and frame length of 10\,ms, the range and velocity resolutions are 0.375\,m and 0.535\,m/s, being small enough.
Based on this precondition, the impact of multi-target coupling is not discussed in this work.
}

\section{Numerical Results}\label{sec.sim}
The simulations in this work aim to investigate the performance limitation, including coverage and accuracy, of sensing in cooperative ISAC networks. Afterward, the approximation error of the simplified CRLB is examined, and its relationship with the proposed criteria in \textbf{Theorem~\ref{t.1}} is evaluated. The system configuration considers the 5G parameters in urban scenarios from Table~\ref{tab:scenario}, and the detailed parameters are given in Table~\ref{tab:para}. Without loss of generality, four sensing nodes arranged at the vertices of a square are included in the simulation (as shown in Fig.~\ref{fig:model}), their coordinates (in meters) are $\mathbf{p}_1=[-50,-50]^T,\mathbf{p}_2=[50,-50]^T,\mathbf{p}_3=[50,50]^T,$ and $\mathbf{p}_4=[-50,50]^T$, thus the \ac{ISD} is 100\,m. Their ULAs are oriented toward the center $[0,0]^T$.
The target is assumed to be a car with RCS of 10\,dBsm. Unless otherwise specified, the default velocity (in km/h) is set to $\mathbf{v}=[30,0]^T$. The signal-level EFIM is calculated by (\ref{eq.sige}) and (\ref{eq.siglvcrlb}) with $B_\text{rms}^2=B^2/12$, $T_\text{rms}^2=T_\text{F}^2/12$.
The ESNR in the simulation is calculated based on the radar propagation equation, i.e., 
\begin{align}
    \gamma=\frac{P_\text{t}N_\text{t}^2\lambda^2\varsigma N_\text{s}}{(4\pi)^3(r^\text{t}r^\text{r})^2\sigma^2},
\end{align}
where for monostatic link, $r^\text{t}=r^\text{r}$. $N_\text{s}=\gamma/\text{SNR}$ denotes the number of samples. For PMCW, the sampling rate is generally the chip rate $T_\text{c}=(1+\alpha)/B$, hence $N_\text{s}=BT_\text{F}/(1+\alpha)$; for OFDM and OTFS, $T_\text{s}=1/B$, $N_\text{s}=BT_\text{F}$ \cite{11570946}. In this work, we assume $N_\text{s}=BT_\text{F}$. For beamforming, \ac{MRT} $\mathbf{w}=\mathbf{b}_{m}(\vartheta_{m})$ is considered, hence the Tx antenna gain is $N_\text{t}^2$.

\begin{table}
    \centering
    \caption{System parameters.}
    \begin{tabular}{c|c}
        \hline
        Parameter & Value \\ \hline
        Carrier frequency $f_\text{c}$ & 28\,GHz \\
        Bandwidth $B$ & 400\,MHz \\
        Frame length $T_\text{F}$ & 10\,ms \\
        Antenna & ULA with $N_\text{t}=N_\text{r}=16$ \\
        Range resolution $\Delta r$ & 0.375\,m / 0.75\,m \\
        Velocity resolution $\Delta v$ & 0.54\,m/s / 1.07\,m/s \\
        Transmission power per element $P_\text{t}$ & 13\,dBm \\
        Noise figure & 10\,dB \\
        RCS $\varsigma$ & 10\,dBsm \\
        \hline
    \end{tabular}
    \label{tab:para}
\end{table}

\subsection{Performance Investigation}\label{sec.perf}

\begin{figure*}[!t]
    \centering
    \subfloat[MMS PEB.]{
        \includegraphics[width=0.175\linewidth]{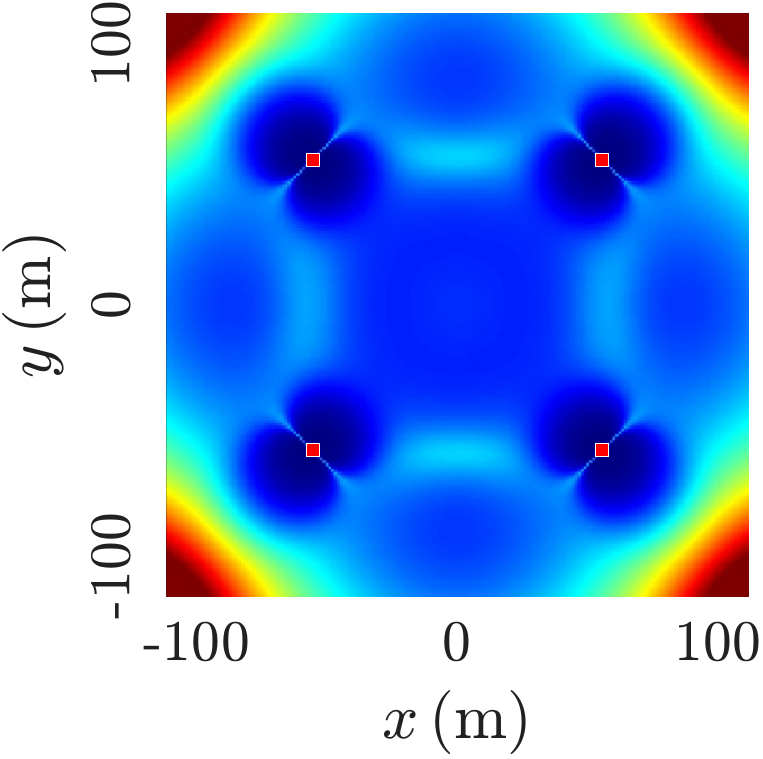}\label{fig.MMSPEB}
    }\hspace{-1mm}
    \subfloat[$1\times3$ MBS PEB.]{
        \includegraphics[width=0.175\linewidth]{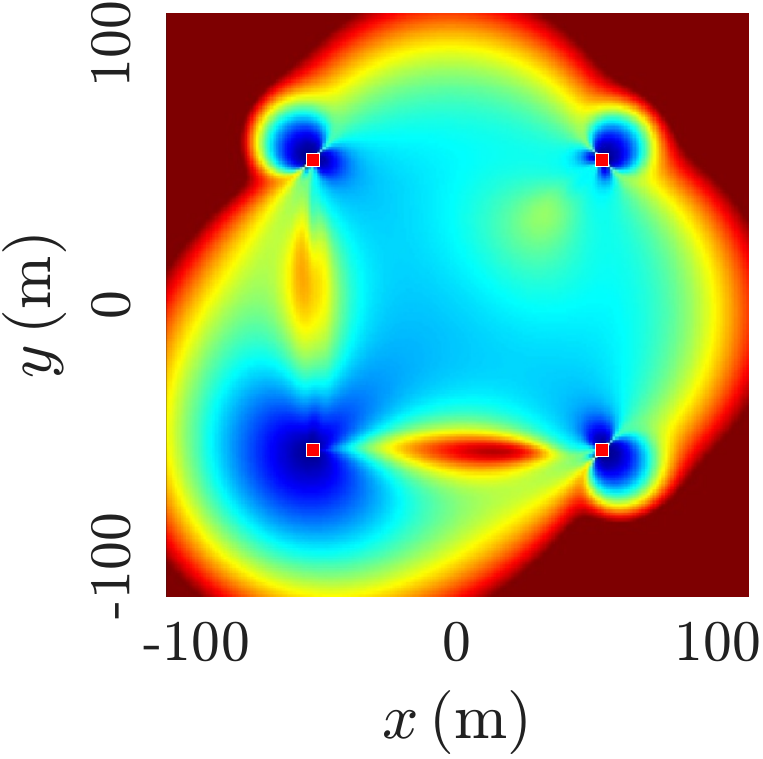}\label{fig.MBSPEB1}
    }\hspace{-1mm}
    \subfloat[$4\times3$ MBS PEB.]{
        \includegraphics[width=0.175\linewidth]{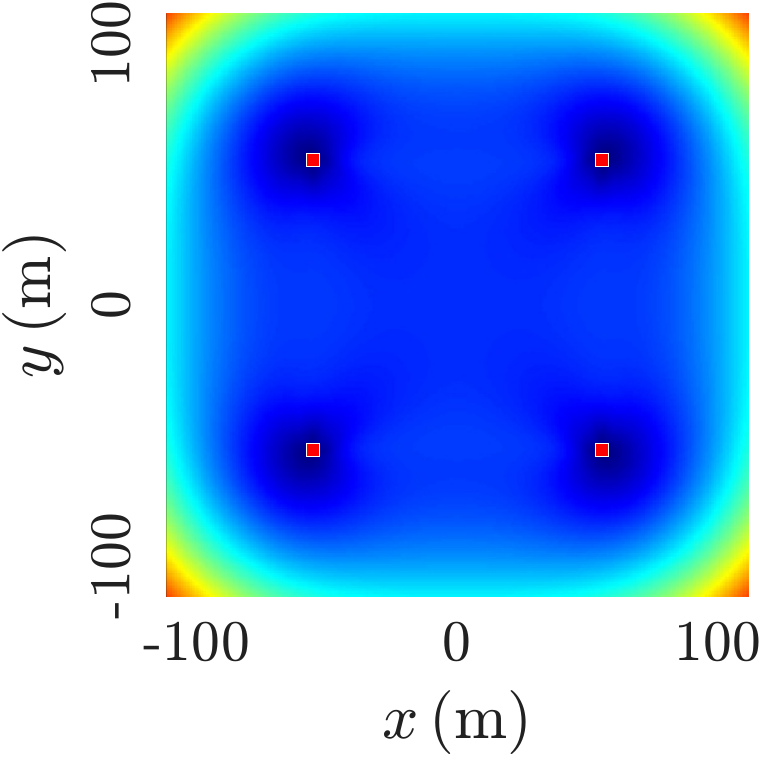}\label{fig.MBSPEB4}
    }\hspace{-1mm}
    \subfloat[$1\times4$ MXS PEB.]{
        \includegraphics[width=0.175\linewidth]{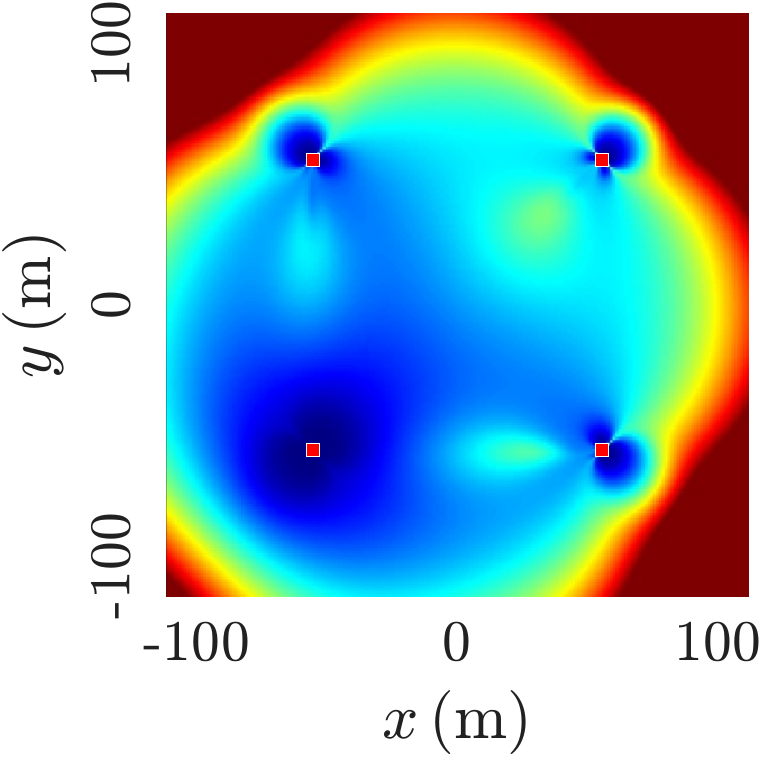}\label{fig.MXSPEB1}
    }\hspace{-1mm}
    \subfloat[$4\times4$ MXS PEB.]{
        \includegraphics[width=0.175\linewidth]{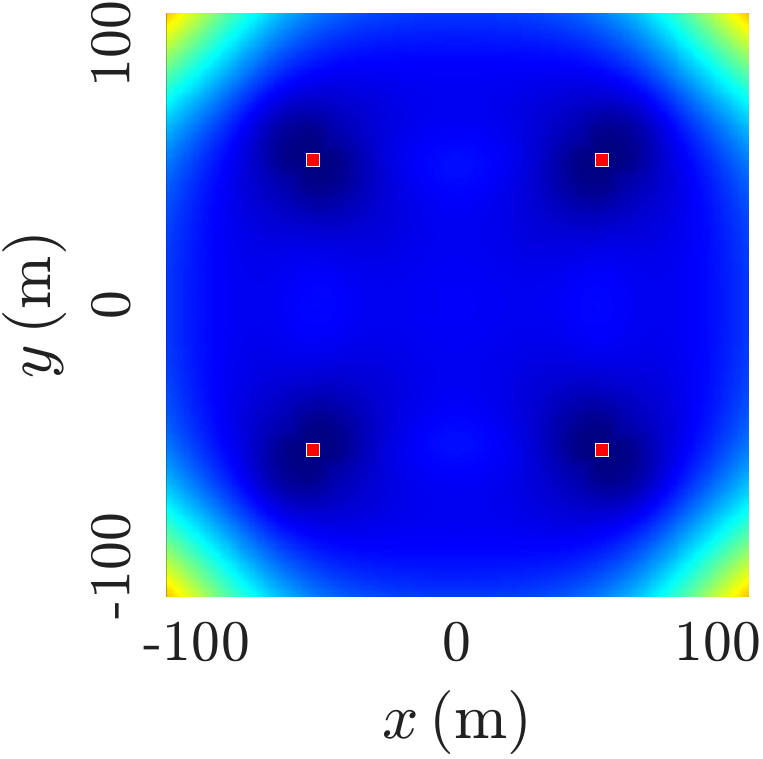}\label{fig.MXSPEB4}
    }
    \subfloat{
        \includegraphics[width=0.04\linewidth]{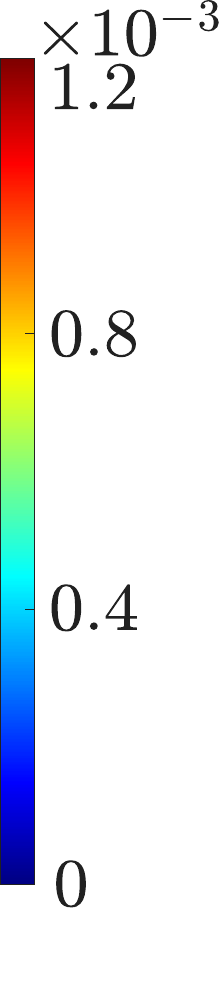}
    }\addtocounter{subfigure}{-1}
    \vspace{-1em}

    \subfloat[MMS VEB.]{
        \includegraphics[width=0.175\linewidth]{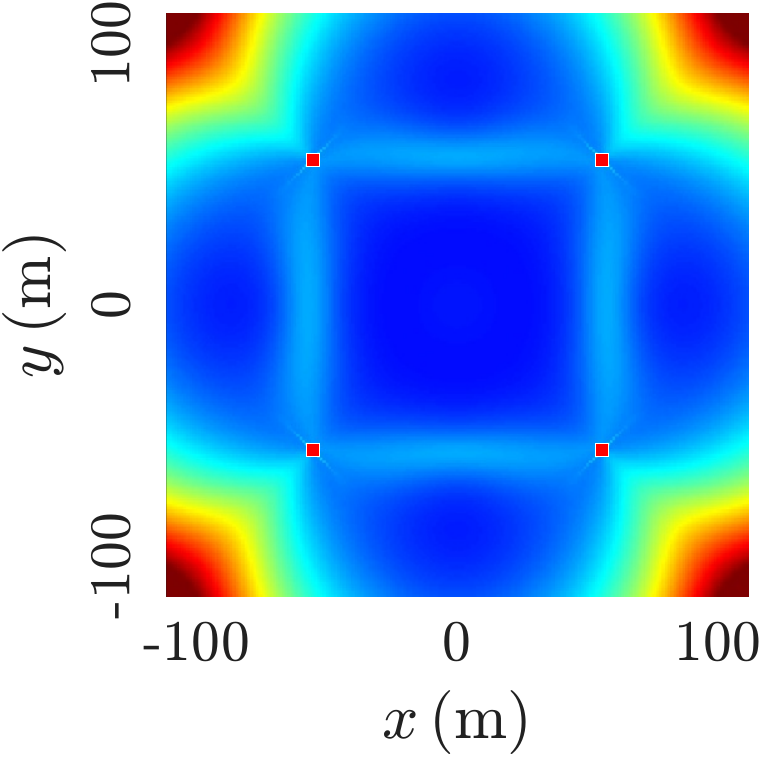}\label{fig.MMSVEB}
    }\hspace{-1mm}
    \subfloat[$1\times3$ MBS VEB.]{
        \includegraphics[width=0.175\linewidth]{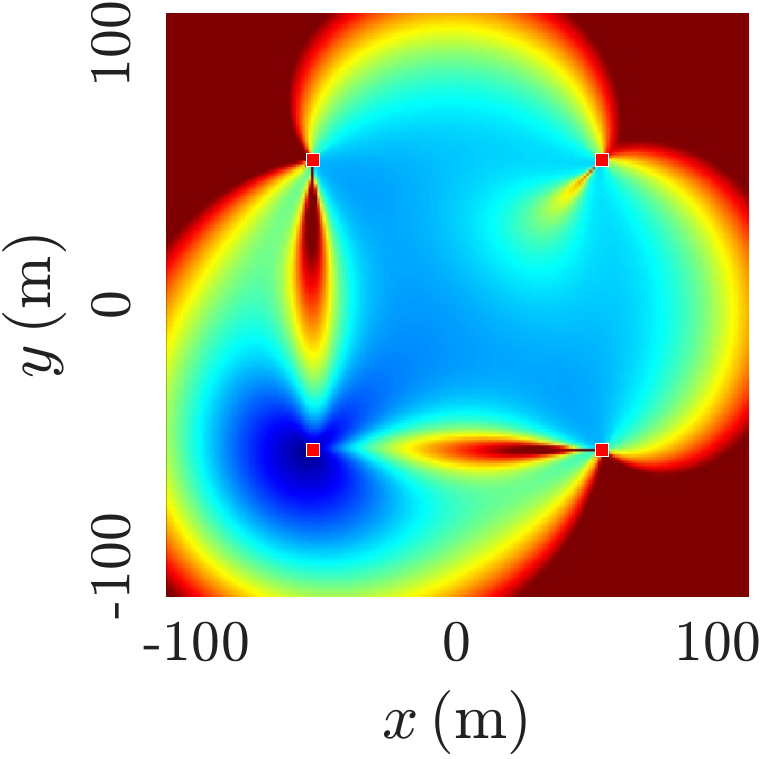}\label{fig.MBSVEB1}
    }\hspace{-1mm}
    \subfloat[$4\times3$ MBS VEB.]{
        \includegraphics[width=0.175\linewidth]{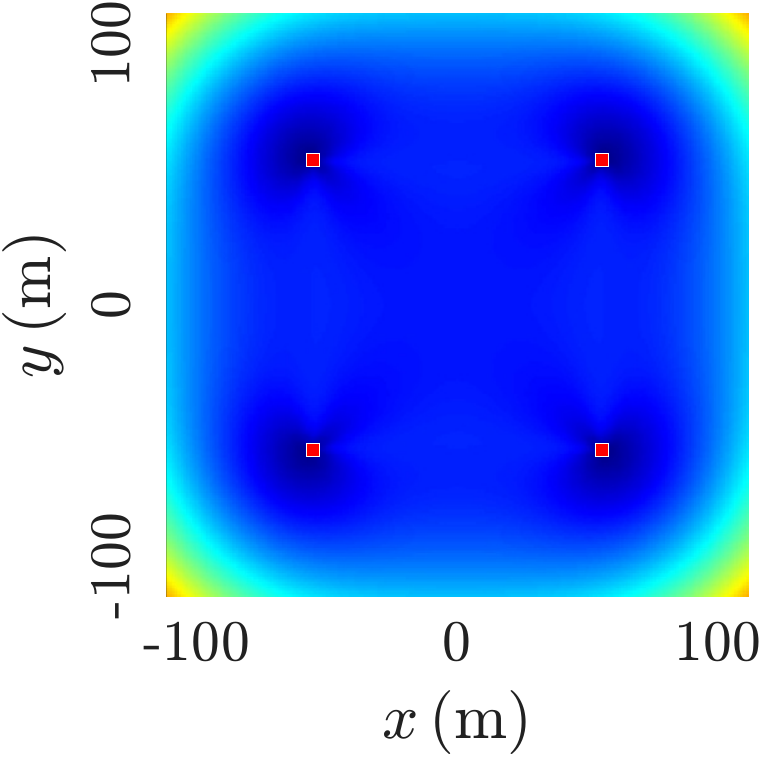}\label{fig.MBSVEB4}
    }\hspace{-1mm}
    \subfloat[$1\times4$ MXS VEB.]{
        \includegraphics[width=0.175\linewidth]{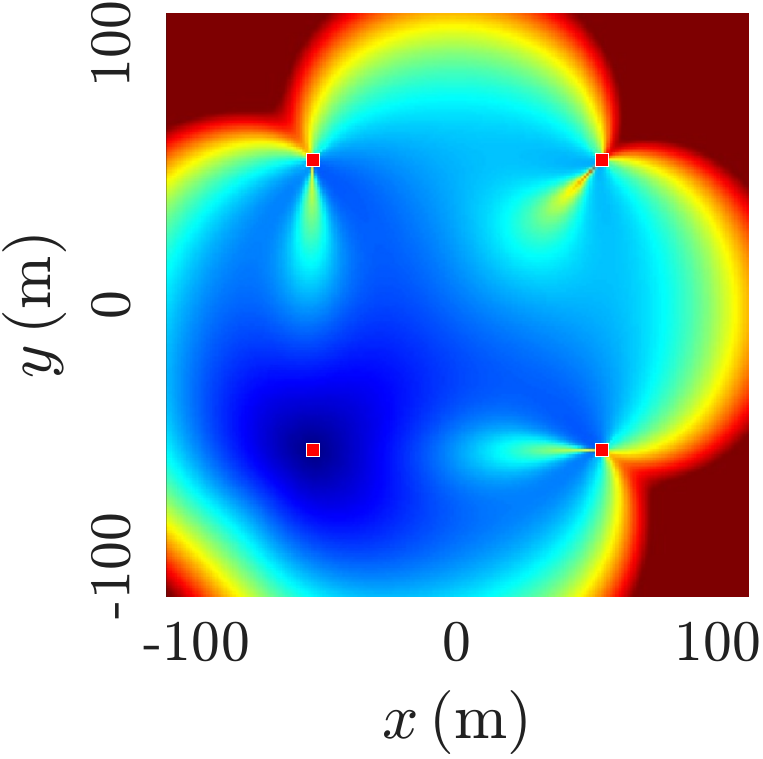}\label{fig.MXSVEB1}
    }\hspace{-1mm}
    \subfloat[$4\times4$ MXS VEB.]{
        \includegraphics[width=0.175\linewidth]{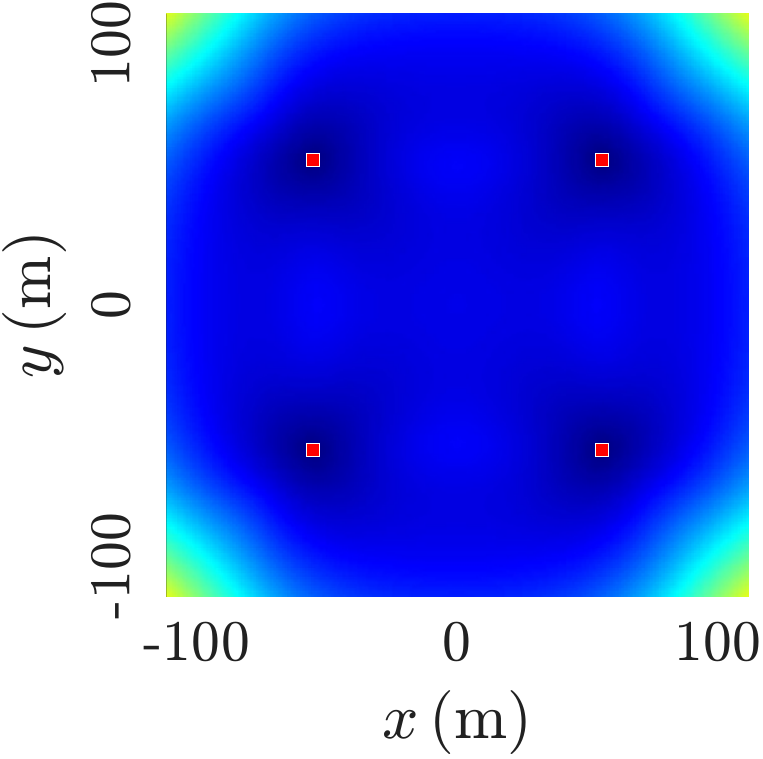}\label{fig.MXSVEB4}
    }
    \subfloat{
        \includegraphics[width=0.04\linewidth]{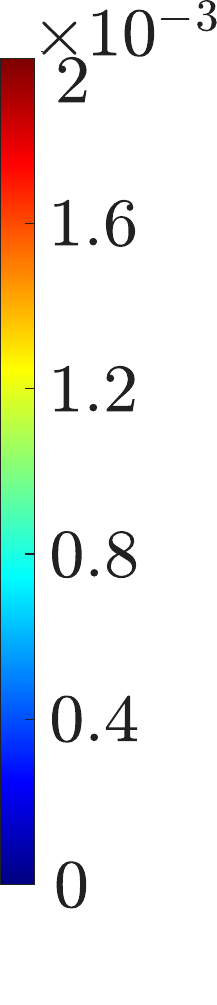}
    }
    \caption{PEB and VEB of different cooperative sensing types calculated by the full CRLB formula. The color denotes the value of PEB in m or VEB in m/s.}
    \label{fig.heatmap}
    \vspace{-1em}
\end{figure*}

In this section, the coverage and accuracy limits of the cooperative sensing network are evaluated via PEB and VEB. The heatmaps in Fig.~\ref{fig.heatmap} illustrate the geometrical characteristics of PEB (first row) and VEB (second row) of different sensing types, where for MBS and MXS, the cases of one transmitter ($1\times3$ and $1\times4$) and four transmitters ($4\times3$ and $4\times4$) are considered, while for MMS, the number of TRxs is always 4. The color bars on the right side indicate the PEB in meters and VEB in m/s on the heatmaps. The four sensing nodes are marked by red points. For MBS and MXS with 1 Tx, the best performance region lies around the Tx, as the overall path loss of all 3 (MBS) and 4 (MXS) links is lowest there. When the number of Txs grows to 4, the corresponding coverage becomes much better. MXS outperforms MBS due to the existence of the additional monostatic links, especially in the areas close to the nodes.

For MMS, the locations on the left and right sides ($\theta\approx\pm90^\circ$) of the ULAs usually have higher PEB due to the absence of angle information since $\cos^2(\pm90^\circ)=0$. Moreover, the PEB and VEB become larger along the edge of the square formed by the sensing nodes due to the geometric degradation: 
As given in (\ref{eq.efimmono}), in the MMS system, all measurement sensitivities are determined by the monostatic \ac{LoS} directions $\mathbf{u}$ (and $\mathbf{u}_\perp$). 
When a target is close to the edge between two nodes, e.g., nodes 1 and 2, these two nodes become the dominant monostatic radars due to lower path loss. However, they cannot provide sufficient direction diversity in $\mathbf{F}'_\text{P}$ and $\mathbf{F}'_\text{V}$ at this position since $\mathbf{u}_1\approx-\mathbf{u}_2$, the obtained position and velocity information mainly lie on one direction $\pm\mathbf{u}_1$. Therefore, the PEB and VEB increase. In contrast, the sensitive directions of bistatic links are jointly determined by $\mathbf{u}^\text{t}$ and $\mathbf{u}^\text{r}$. This provides additional spatial diversity and mitigates the edge-induced degradation. 

\begin{figure}
    \centering
    \subfloat[PEB, VEB vs. $P_\text{t}$.]{
        \includegraphics[width=0.47\linewidth]{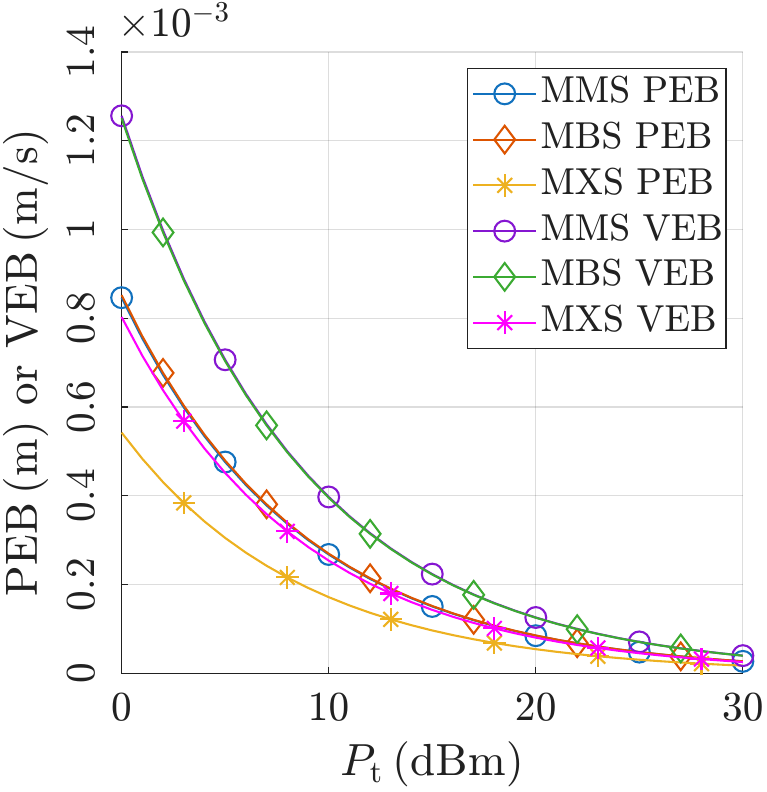}\label{fig.power}
    }\hspace{-1mm}
    \subfloat[PEB, VEB vs. $N_\text{r}$.]{
        \includegraphics[width=0.47\linewidth]{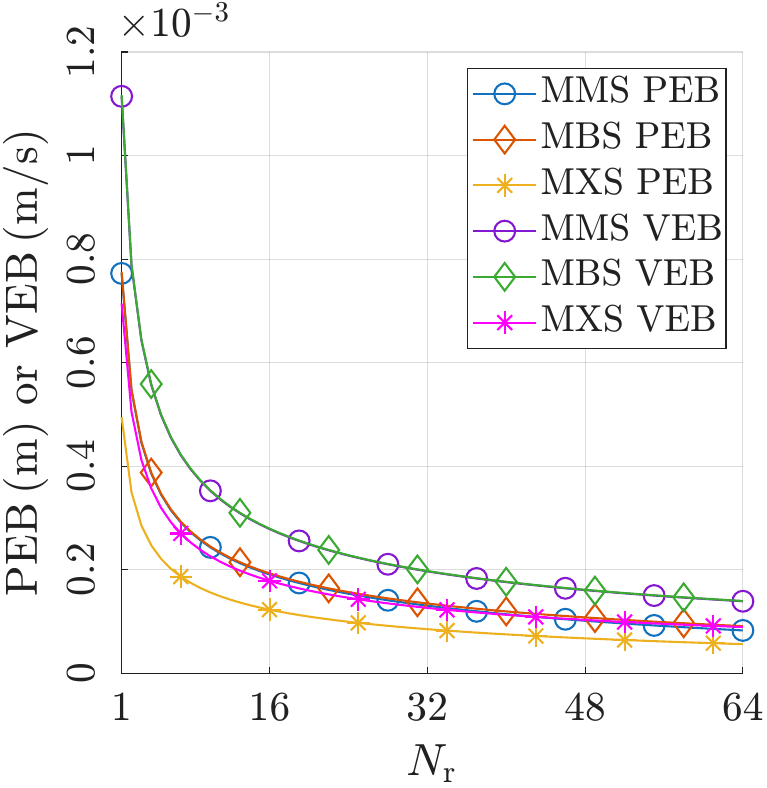}\label{fig.antenna}
    }

    \subfloat[PEB, VEB vs. $B$.]{
        \includegraphics[width=0.47\linewidth]{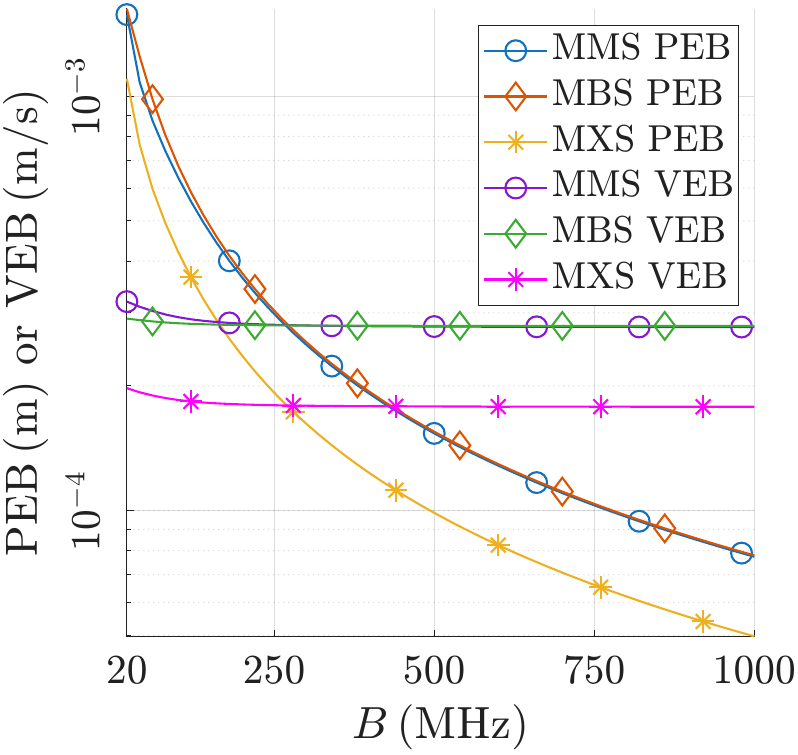}\label{fig.bw}
    }\hspace{-1mm}
    \subfloat[PEB, VEB vs. $T_\text{F}$.]{
        \includegraphics[width=0.47\linewidth]{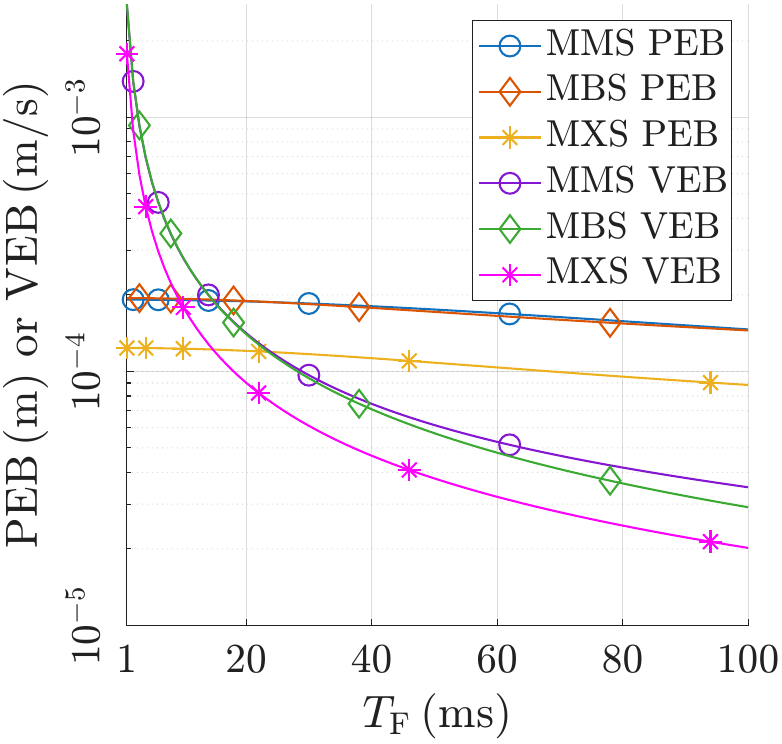}\label{fig.tf}
    }

    \subfloat[PEB, VEB vs. ISD.]{
        \includegraphics[width=0.47\linewidth]{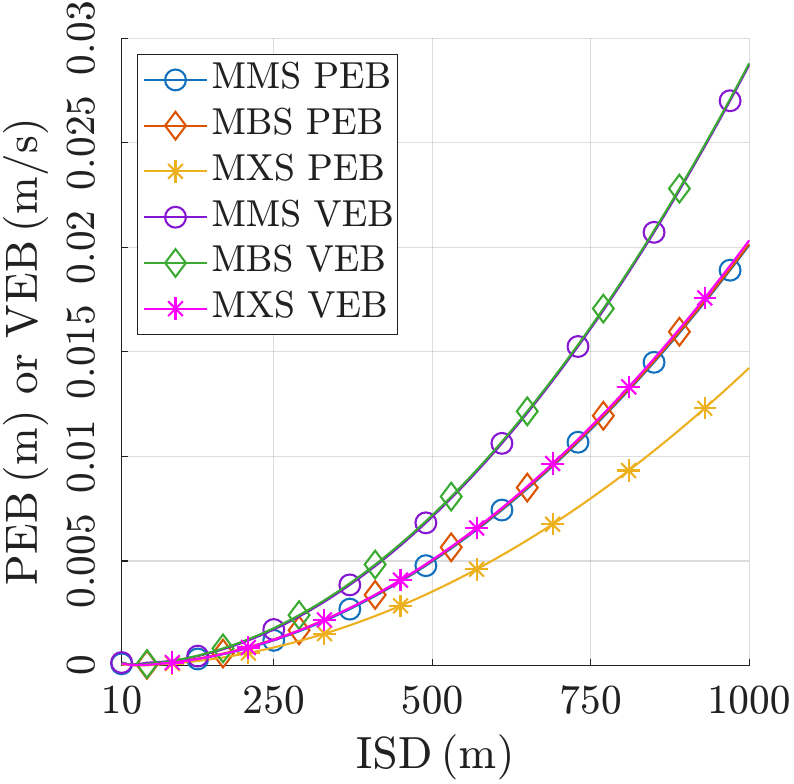}\label{fig.isd}
    }\hspace{-1mm}
    \subfloat[PEB, VEB vs. $|\mathbf{v}|$.]{
        \includegraphics[width=0.47\linewidth]{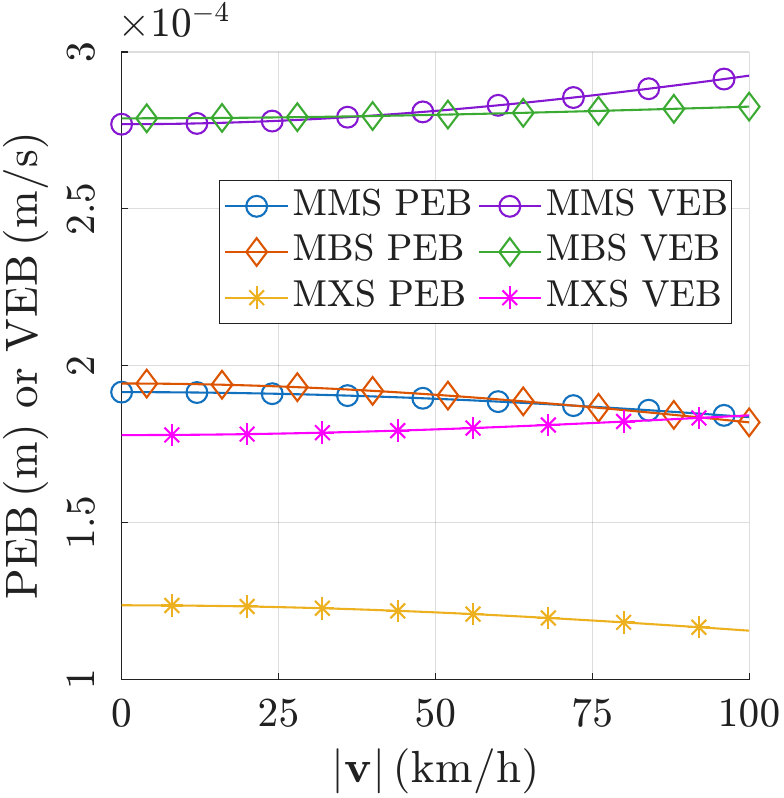}\label{fig.v}
    }
    \caption{PEB and VEB versus resource or system parameter curves. The results of MMS, $4\times3$ MBS, and $4\times4$ MXS are depicted.}
    \label{fig.performance}
\end{figure}

The influence of resource allocation and system configuration on PEB and VEB is described in Fig.~\ref{fig.performance}, where except for the parameter being tested, other parameters are set to the default values, and it is assumed that all the sensing nodes have the same resource allocation. The target is located at $\mathbf{p}=[25,15]^T$\,m with a velocity of $\mathbf{v}=[30,0]^T$\,km/h.
In addition, although ESNR $\gamma=BT_\text{F}\text{SNR}$ changes with bandwidth and frame length, we assume a fixed ESNR in Fig.~\ref{fig.performance}(c) and (d) in order to directly reflect the impact of changed resource amounts on PEB and VEB, while the result with changed ESNR/SNR can be found Fig.~\ref{fig.performance}(a).
It is obvious that the PEB and VEB decrease with the increasing transmission power and antenna size, since these resources are positively related to $E_r, E_{v_\text{r}}$, and $E_\theta$ in (\ref{eq.sige}).
Although bandwidth and frame length are mainly relevant to $C_r$ and $C_{v_\text{r}}$ and significantly impact the PEB and VEB, respectively, they also introduce minor mutual influences on VEB and PEB through the components of $\pmb{\Delta}_\text{P}$ and $\pmb{\Delta}_\text{V}$. 
The ISD changes PEB and VEB primarily due to the increased path loss. With increased velocity, the VEB increases since a higher velocity corresponds to a larger $\pmb{\Delta}_\text{P}$, while the PEB is slightly reduced since $\pmb{\Delta}_\text{V}\propto |\mathbf{v}|^2$. 
Additionally, it is illustrated that the PEB and VEB of MMS and $4\times3$ MBS are almost the same for the tested position, while MXS combines MMS and MBS and therefore benefits from the best achievable accuracy.

\subsection{Validation of the Simplified CRLB and the Proposed Criterion}\label{sec.dev}

\begin{figure*}
    \centering
    \subfloat[MMS PEB.]{
        \includegraphics[width=0.175\linewidth]{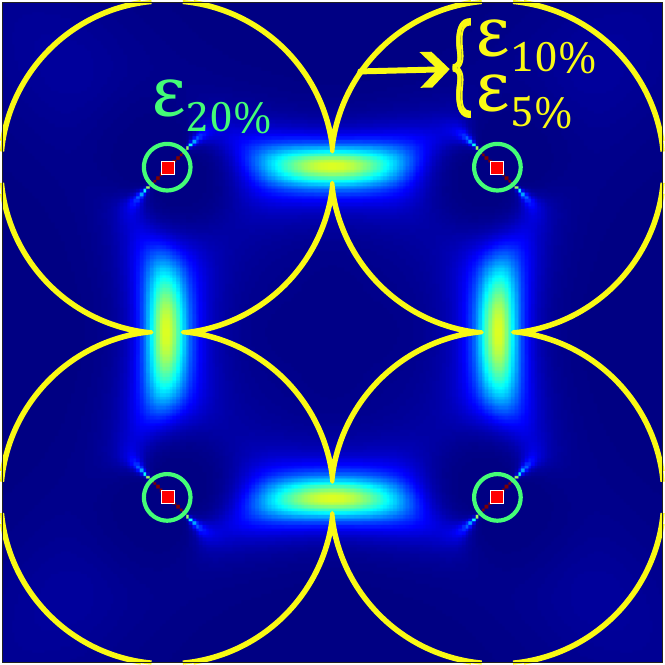}\label{fig.MP}
    }\hspace{-1mm}
    \subfloat[$1\times3$ MBS PEB.]{
        \includegraphics[width=0.175\linewidth]{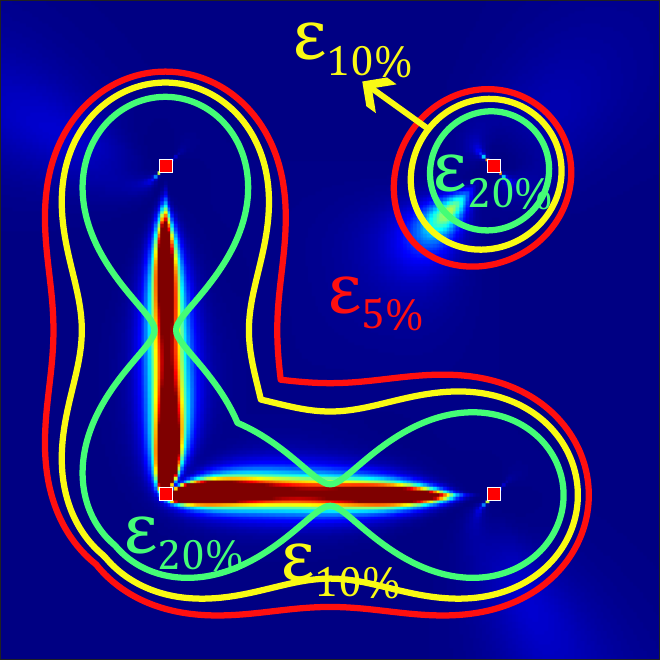}\label{fig.BP1}
    }\hspace{-1mm}
    \subfloat[$4\times3$ MBS PEB.]{
        \includegraphics[width=0.175\linewidth]{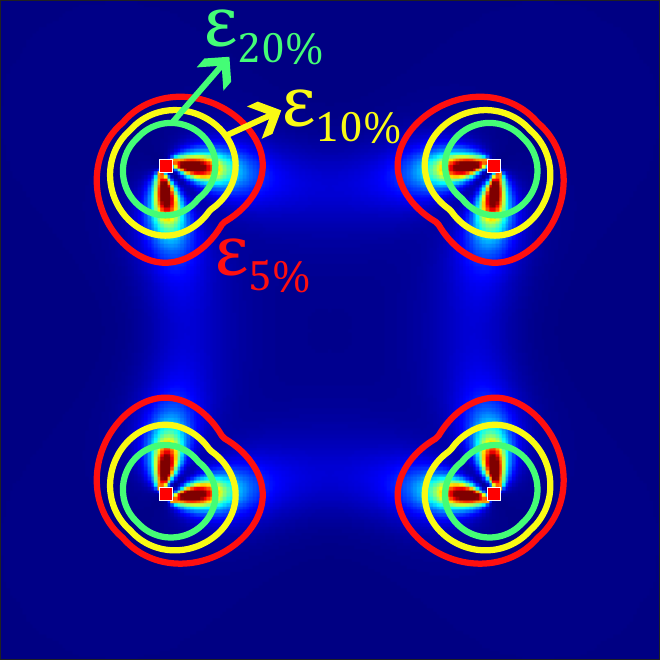}\label{fig.BP4}
    }\hspace{-1mm}
    \subfloat[$1\times4$ MXS PEB.]{
        \includegraphics[width=0.175\linewidth]{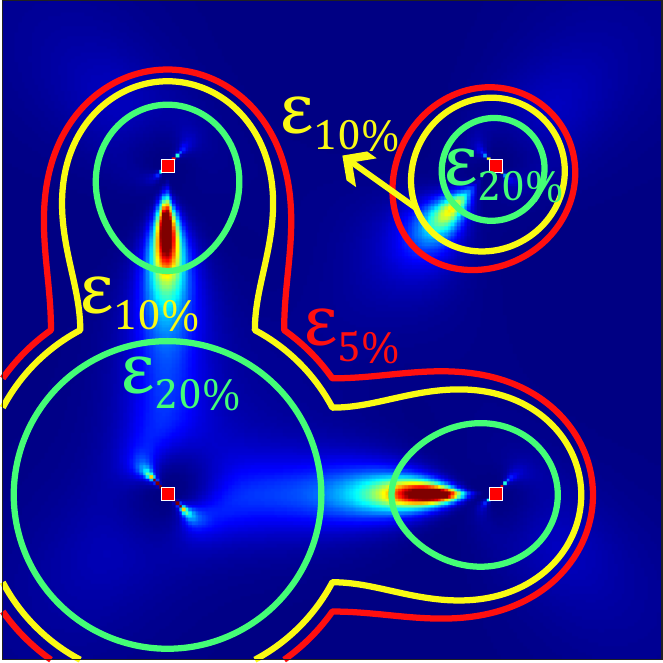}\label{fig.XP1}
    }\hspace{-1mm}
    \subfloat[$4\times4$ MXS PEB.]{
        \includegraphics[width=0.175\linewidth]{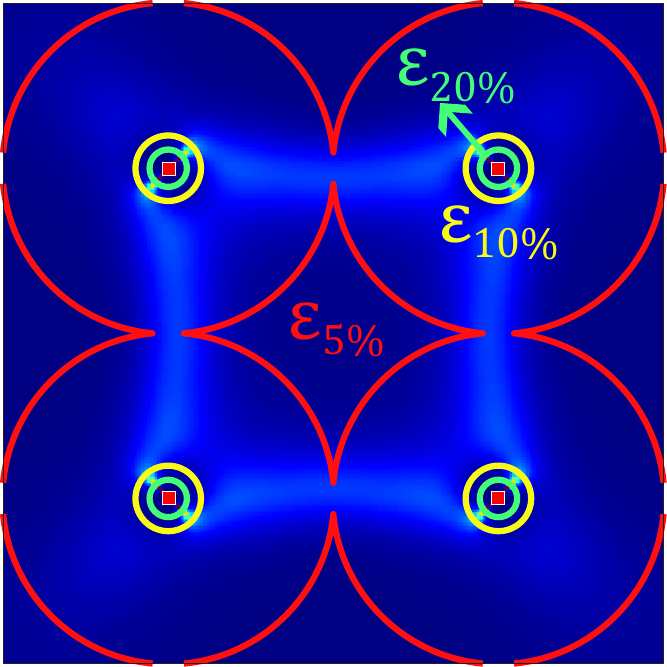}\label{fig.XP4}
    }
    \subfloat{
        \includegraphics[width=0.04\linewidth]{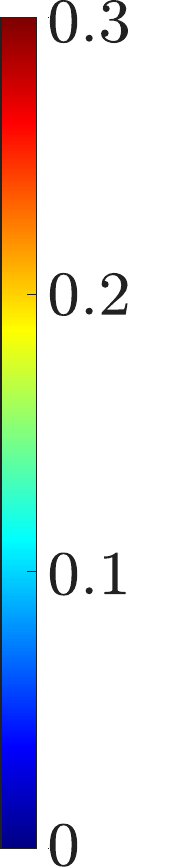}
    }\addtocounter{subfigure}{-1}
    \vspace{-1em}

    \subfloat[MMS VEB.]{
        \includegraphics[width=0.175\linewidth]{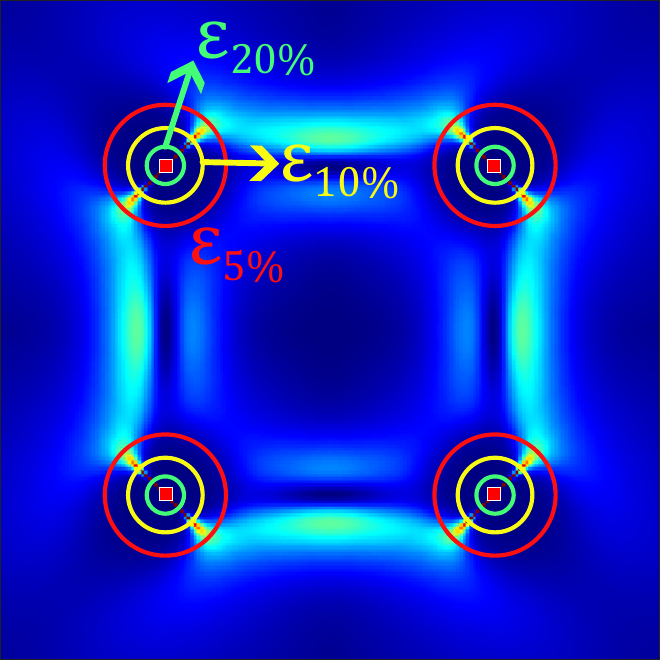}\label{fig.MV}
    }\hspace{-1mm}
    \subfloat[$1\times3$ MBS VEB.]{
        \includegraphics[width=0.175\linewidth]{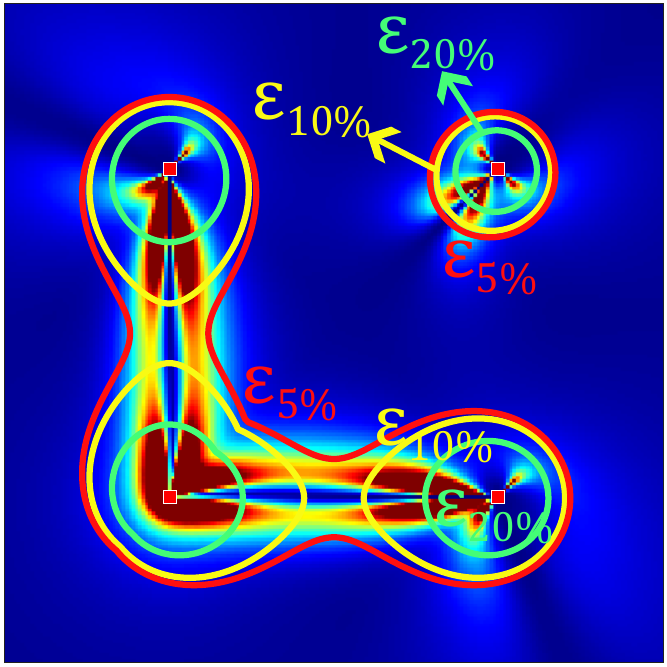}\label{fig.BV1}
    }\hspace{-1mm}
    \subfloat[$4\times3$ MBS VEB.]{
        \includegraphics[width=0.175\linewidth]{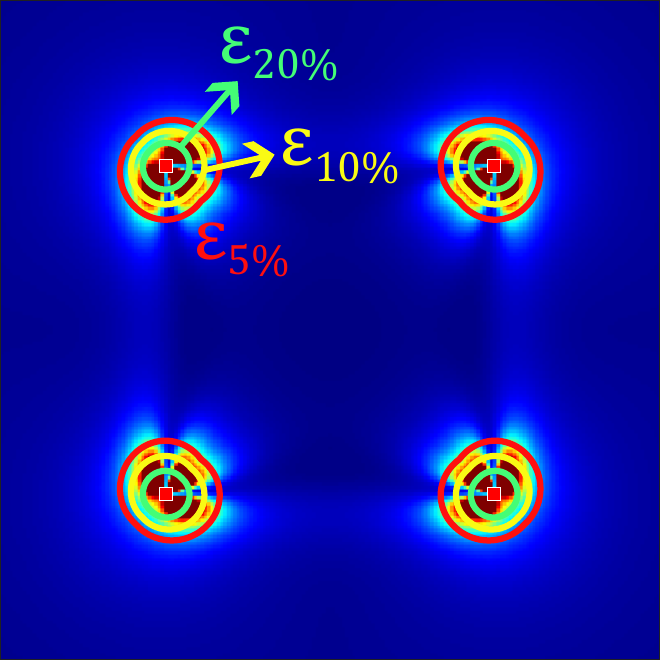}\label{fig.BV4}
    }\hspace{-1mm}
    \subfloat[$1\times4$ MXS VEB.]{
        \includegraphics[width=0.175\linewidth]{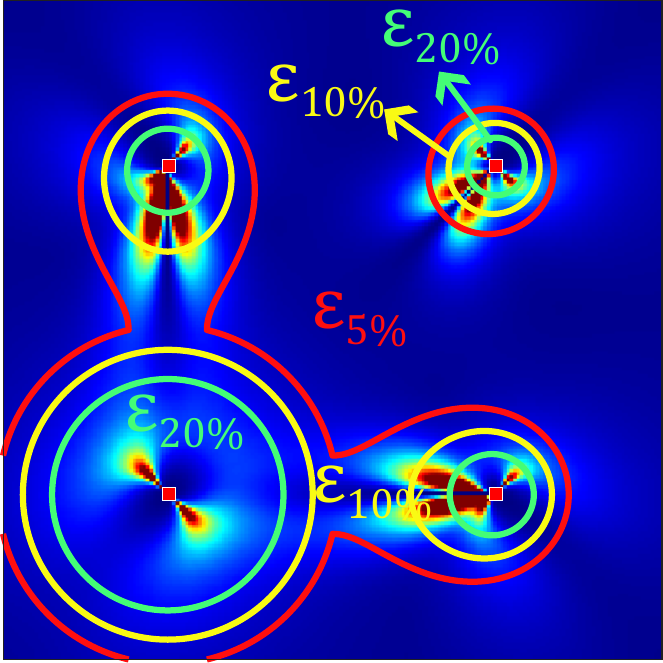}\label{fig.XV1}
    }\hspace{-1mm}
    \subfloat[$4\times4$ MXS VEB.]{
        \includegraphics[width=0.175\linewidth]{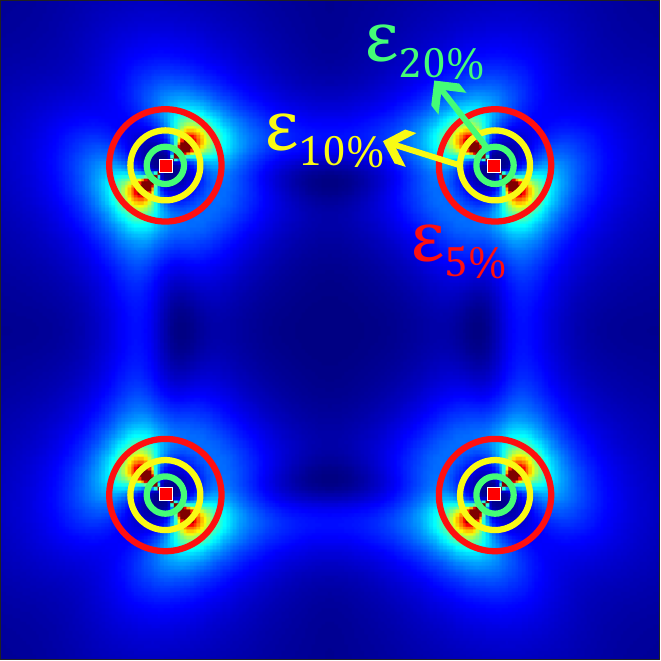}\label{fig.XV4}
    }
    \subfloat{
        \includegraphics[width=0.04\linewidth]{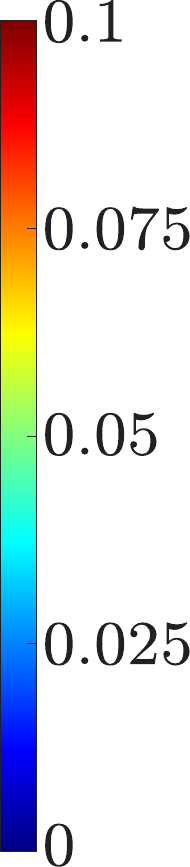}
    }
    \caption{$\delta_\text{P}$ and $\delta_\text{V}$ of different cooperative sensing modes. The heatmap color denotes the NMAE. The contours denote the positions corresponding to $\varepsilon_{20\%}$, $\varepsilon_{10\%}$, and $\varepsilon_{5\%}$. The coordinate configuration is identical to Fig.~\ref{fig.heatmap} and thus is omitted for better visibility.}
    \label{fig.heatmap1}
    \vspace{-1em}
\end{figure*}

\begin{table}[t]
\centering
\caption{$\varepsilon_{20\%}$, $\varepsilon_{10\%}$, and $\varepsilon_{5\%}$ of different sensing modes. $r_{20\%},r_{10\%},r_{5\%}$ in meters are the corresponding ranges $r$ or $\sqrt{r^\mathrm{t}r^\mathrm{r}}$.}
\begin{tabular}{c|c|c|c|c|c|c|c}
\hline
 & & $\varepsilon_{20\%}$ & $r_{20\%}$ & $\varepsilon_{10\%}$ & $r_{10\%}$ & $\varepsilon_{5\%}$ & $r_{5\%}$\\
\hline
\multirow{2}{*}{MMS}
& P  & 1.21 & 7.06 & 8.61 & 50.23 & 8.68 & 50.63\\
& V & 0.97 & 5.66 & 1.94 & 11.32 & 3.15 & 18.38 \\
\hline
\multirow{2}{*}{$1\negthinspace\times\negthinspace3$ MBS}
& P & 8.59 & 50.10 & 9.60 & 56.06 & 10.36 & 60.43 \\
& V & 7.14 & 41.65 & 8.43 & 49.18 & 8.81 & 51.39 \\
\hline
\multirow{2}{*}{$4\negthinspace\times\negthinspace3$ MBS}
& P & 6.15 & 35.87 & 7.02 & 40.95 & 7.82 & 45.62 \\
& V & 4.47 & 26.07 & 5.54 & 32.31 & 6.36 & 37.10 \\
\hline
\multirow{2}{*}{$1\negthinspace\times\negthinspace4$ MXS}
& P & 8.00 & 46.66 & 9.70 & 56.58 & 10.50 & 61.25 \\
& V & 6.02 & 35.12 & 7.54 & 43.98 & 8.81 & 51.39 \\
\hline
\multirow{2}{*}{$4\negthinspace\times\negthinspace4$ MXS}
& P & 0.97 & 5.66 & 1.70 & 9.92 & 8.61 & 50.23\\
& V & 0.97 & 5.66 & 1.82 & 10.62 & 2.92 & 17.03 \\
\hline
\end{tabular}
\label{tab.eps}
\end{table}

This section evaluates the error of the simplified CRLB and its relationship with $\varepsilon=\sqrt{1/\epsilon}$ in (\ref{eq.crinode}) and (\ref{eq.crinodeb}), where $\varepsilon$ equals $\min_{m} \Bar{r}_m/\Bar{v}$ for MMS and $\min_{m}\sqrt{\Bar{r}_m^\text{t}\Bar{r}_m^\text{r}}/\Bar{v}$ for MBS, directly reflecting the ratio between range and velocity. 
The default parameters are the same as previously defined. The \ac{NMAE} is employed to express the error produced by the simplification, which calculates the ratio between the simplification-induced error and the full PEB and VEB:
\begin{align}
    \delta_\text{P}=\frac{|\text{PEB}-{\text{PEB}_0}|}{\text{PEB}},\quad \delta_\text{V}=\frac{|\text{VEB}-{\text{VEB}_0}|}{\text{VEB}},
\end{align}
where $\text{PEB}_0$ and $\text{VEB}_0$ are the simplified results based on $\mathbf{F}_{\text{P}0}$ and $\mathbf{F}_{\text{V}0}$. We first analyze the NMAEs geometrically via the heatmap. Since the target (e.g., vehicle) may drive in arbitrary directions, we evaluate the results for velocity directions varied from 0 to $2\pi$ with a step size of $\pi/12$, and the maximum NMAE at each position is extracted and shown in the heatmap in Fig.~\ref{fig.heatmap1}. Further, we depict three contours corresponding to $\varepsilon_{20\%}$, $\varepsilon_{10\%}$, and $\varepsilon_{5\%}$ on each heatmap to check the validity and tightness of the proposed criterion in \textbf{Theorem~\ref{t.1}}, where $\varepsilon_{X\%}$ corresponds to the sufficient (but not necessary) condition boundary for $\delta_\text{P}$ or $\delta_\text{V}<X\%$, i.e., the NMAE outside this contour is guaranteed to be below $X\%$, while the error inside this boundary may be either below or above $X\%$:
\begin{align}
    \varepsilon_{X\%}^\text{P}= \max_{\delta_\text{P}} \min_m \Bar{r}_m/\Bar{v},\ \  \varepsilon_{X\%}^\text{V}= \max_{\delta_\text{V}} \min_m \Bar{r}_m/\Bar{v}.
\end{align}
The result is depicted in Fig.~\ref{fig.heatmap1}, where the map color represents the NMAE, while the contours of $\varepsilon_{20\%}$, $\varepsilon_{10\%}$, and $\varepsilon_{5\%}$ are marked in green, yellow, and red, whose values are listed in Table~\ref{tab.eps}, where the corresponding ranges $r_{20\%},r_{10\%}, \text{ and } r_{5\%}$ are also recorded for reference.

It is clear that $\delta_\text{P}$ and $\delta_\text{V}$ are sufficiently low (under 10\% or 5\% or even much less) for most positions, particularly for those inside the square formed by the nodes, which generally corresponds to the cooperative sensing region. This demonstrates the applicability of the simplified CRLB.
However, besides the location near sensing nodes, the NMAEs on the edges are more or less higher than other positions, especially for MMS and single-Tx MBS. This phenomenon is produced by geometric degradation: 
As explained in Section~\ref{sec.perf}, when a target is close to the edge between nodes 1 and 2, the dominant range and velocity information $\mathbf{F}_{\text{P}0}$ and $\mathbf{F}_{\text{V}0}$ primarily lie in the direction of $\mathbf{u}_1\approx-\mathbf{u}_2$. In contrast, $\pmb{\Delta}_\text{V}$ provides information in the perpendicular direction $\mathbf{u}_{\perp,1}\approx-\mathbf{u}_{\perp,2}$, thus ignoring its influence implies abandoning the information in the second dimension, resulting in increased NMAE. From another aspect, the condition $\epsilon \mathbf{F}_{\mathrm{P}0}-\mathbf{T}_\mathrm{V}\succeq0$ in \textbf{Proposition~\ref{po.1}} may not hold at edge since $\mathbf{x}^T(\epsilon \mathbf{F}_{\mathrm{P}0}-\mathbf{T}_\mathrm{V})\mathbf{x}$ could be less than 0 when $\mathbf{x}=\pm \mathbf{u}_{\perp,1}$, violating the proposed sufficient condition.
Although angle measurement can mitigate this effect, as analyzed in Section~\ref{sec.simpmms}, the horizontal antenna size of $16$ in the simulation is too small to completely mitigate the loss due to the ignored $\mathbf{T}_\text{V}$.

For MMS, besides at the edges, the high NMAEs mainly exist near the nodes on the left and right sides of the ULAs, where the angle information of the dominant node is nearly zero. Since at these positions, the angle information matrix and $\mathbf{T}_\text{V}$ are approximately proportional to $1/(r_m^\text{r})^2$ of the dominant node, while the range information is irrelevant to it, the disappearance of angle information makes the criterion of \textbf{Theorem~\ref{t.1}} much tighter, implying a high NMAE.
On the contrary, multi-Tx MBS and MXS can mitigate the significantly increased NMAE at the middle of edges by the increased spatial diversity, while the NMAE is high at the edge endpoints since $\mathbf{g}_m={v_{\perp m}^\text{t}\mathbf{u}_{\perp m}^\text{t}}/{r^\text{t}_m}+{v_{\perp m}^\text{r}\mathbf{u}_{\perp m}^\text{r}}/{r^\text{r}_m}$ grows extremely for positions near the Tx (small $r^\text{t}_m$) or Rx (small $r^\text{t}_m$), the validity of $\epsilon\mathbf{F}_{\mathrm{P}0}-\mathbf{T}_\mathrm{V}\succeq0$ is challenged. 

The tightness of the proposed criterion is also examined. The results show favorable tightness, especially for MBS and the VEB of MXS, where the contours well enclose the high-NMAE regions near the sensing nodes and along the edges. 
In contrast, for MMS, although the contours of $\delta_\text{V}$ and the $\varepsilon_{20\%}$ contour of $\delta_\text{P}$ capture the high NMAE near the nodes in a good manner, the $\varepsilon_{10\%}$ and $\varepsilon_{5\%}$ $\delta_\text{P}$ contours show a broken tightness.
The reason is that the contour corresponding to the proposed criterion for MBS in (\ref{eq.crinodeb}) results in a composite Cassini oval induced by multiple bistatic links, which captures the high NMAE along the node connection lines (especially for the edge endpoints). However, the criterion for MMS in (\ref{eq.crinode}) results in a complement of the union of node-centered circles, 
while the high NMAE at the middle of edges extremely increases the radius of circles. 
\textbf{Theorem~\ref{t.1}} is based on sufficient spatial diversity and approximately isotropic FIM, which is not well-held at the edges, especially for MMS.
Nevertheless, the low-NMAE center area inside the square is always indicated by \textbf{Theorem~\ref{t.1}} regardless of the sensing types, since the underlying assumption holds for the center area.

\begin{figure}
    \centering
    \subfloat[$\delta_\text{P}, \delta_\text{V}$ vs. $P_\text{t}$.]{
        \includegraphics[width=0.47\linewidth]{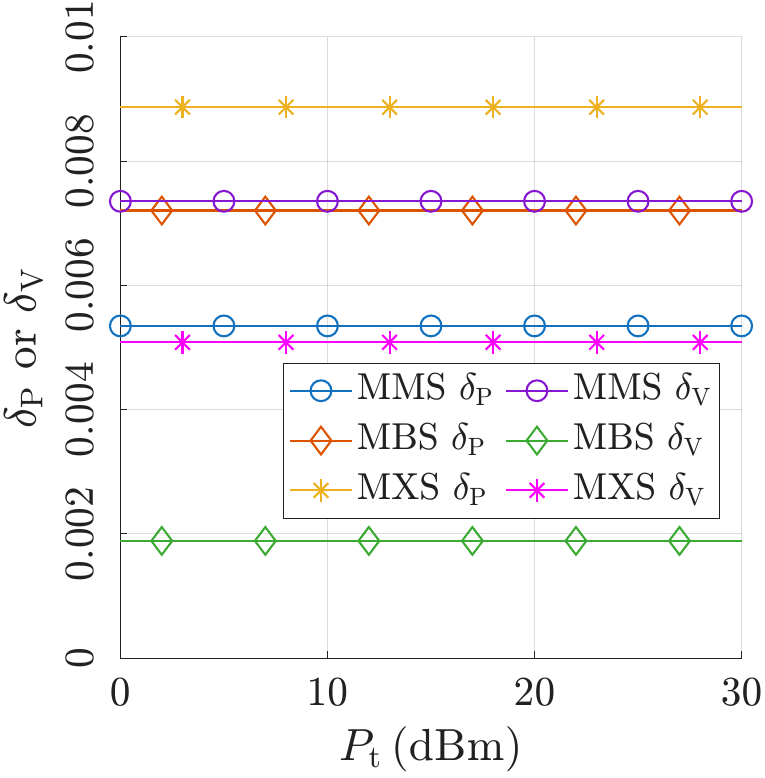}\label{fig.power_d}
    }\hspace{-1mm}
    \subfloat[$\delta_\text{P}, \delta_\text{V}$ vs. $N_\text{r}$.]{
        \includegraphics[width=0.47\linewidth]{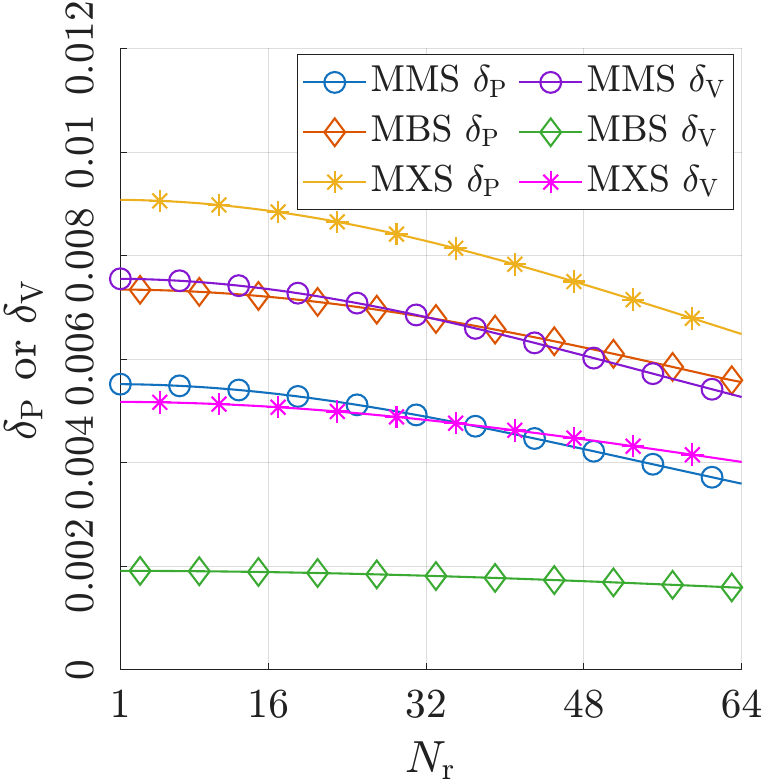}\label{fig.power_ant}
    }\vspace{-0.5em}

    \subfloat[$\delta_\text{P}, \delta_\text{V}$ vs. $B$.]{
        \includegraphics[width=0.47\linewidth]{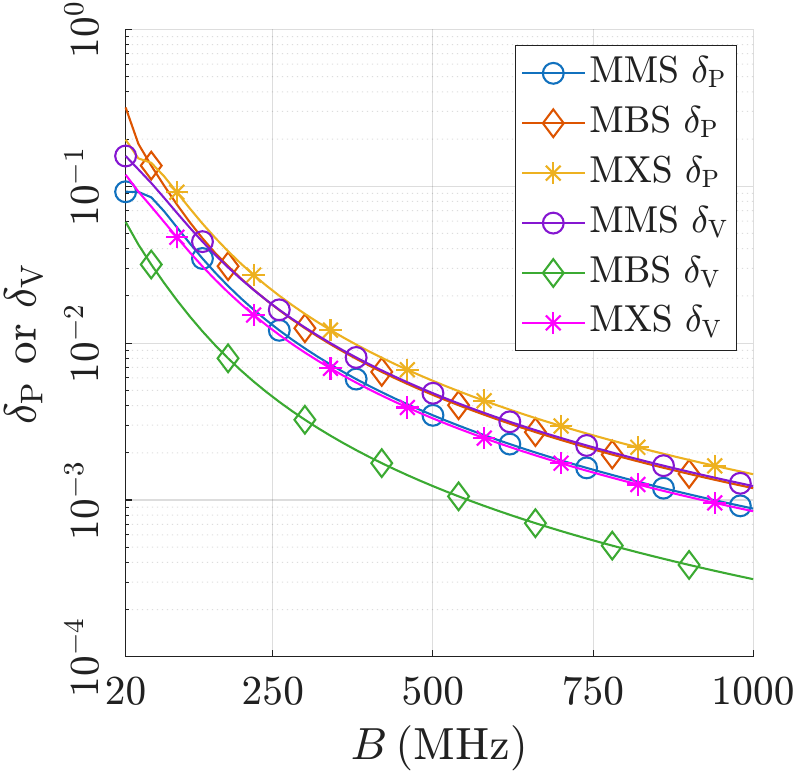}\label{fig.power_bw}
    }\hspace{-1mm}
    \subfloat[$\delta_\text{P}, \delta_\text{V}$ vs. $T_\text{F}$.]{
        \includegraphics[width=0.47\linewidth]{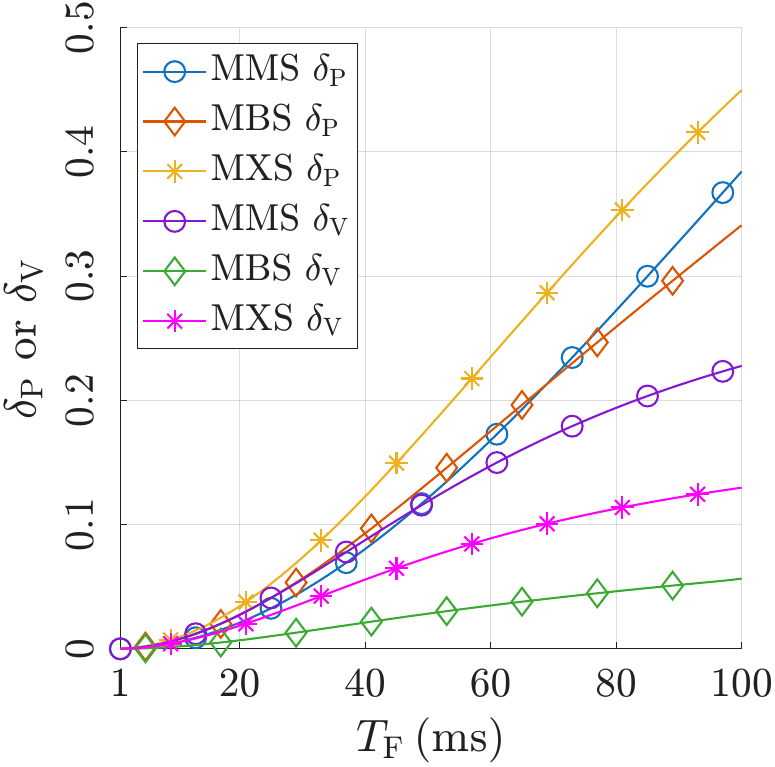}\label{fig.power_tf}
    }\vspace{-0.5em}

    \subfloat[$\delta_\text{P}, \delta_\text{V}$ vs. ISD.]{
        \includegraphics[width=0.47\linewidth]{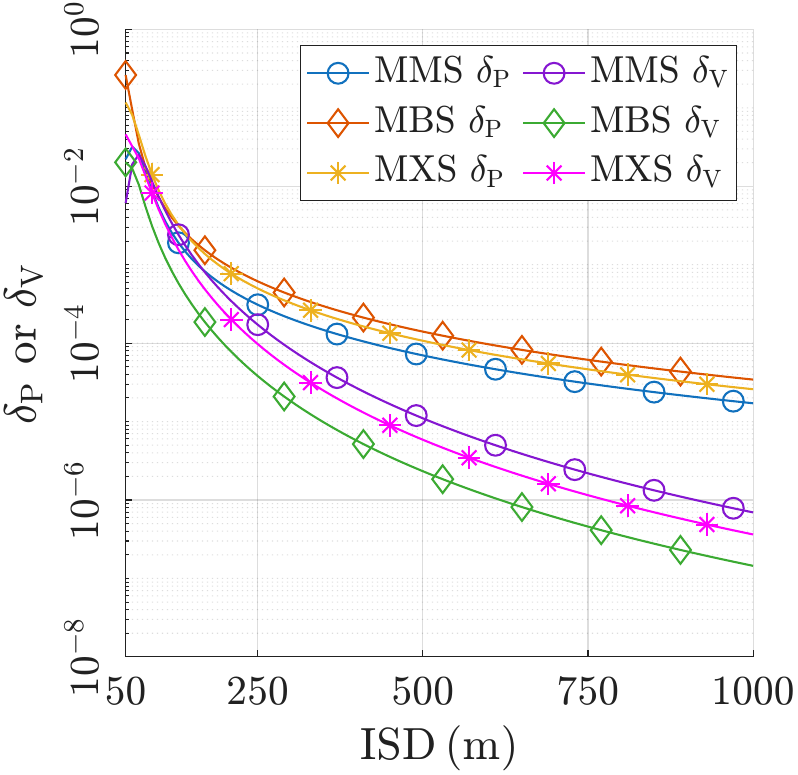}\label{fig.power_isd}
    }\hspace{-1mm}
    \subfloat[$\delta_\text{P}, \delta_\text{V}$ vs. $|\mathbf{v}|$.]{
        \includegraphics[width=0.47\linewidth]{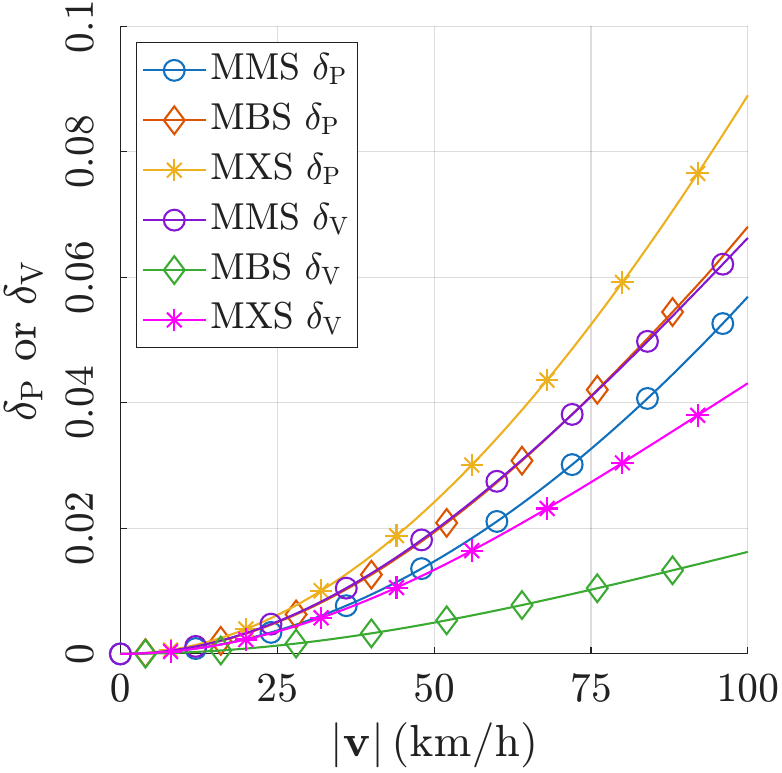}\label{fig.power_v}
    }
    \caption{$\delta_\text{P}\text{ and } \delta_\text{V}$ versus resource or system parameter curves. The results of MMS, $4\times3$ MBS, and $4\times4$ MXS are depicted.}
    \label{fig.dev}
    \vspace{-1em}
\end{figure}

The impact of resource and system parameters on $\delta_\text{P}\text{ and } \delta_\text{V}$ is illustrated in Fig.~\ref{fig.dev}, where the target is located at $\mathbf{p}=[25,15]^T$\,m with a velocity of $|\mathbf{v}|=30$\,km/h. Again, the moving direction varies from 0 to $2\pi$, and the maximum NMAE over all velocity directions is extracted. The transmission power does not influence the NMAEs since both full and simplified PEBs and VEBs are inversely proportional to $\sqrt{\gamma}$. A similar conclusion also applies to $N_\text{r}$, where $E_r,E_{v_\text{r}}\propto N_\text{r}$, while the AoA-related term $E_\theta\propto N_\text{r}^3$ generally has much less contribution compared to others, as explained in Section~\ref{sec.simpmms}. Therefore, the NMAEs only change slightly with $N_\text{r}$. The increased bandwidth reduces the gap between the full and simplified CRLBs by increasing the relative contribution of $\mathbf{F}_{\text{P}0}$ in $\mathbf{F}_\text{P}'$ and decreasing the relative contribution of $\pmb{\Delta}_\text{P}$ in $\mathbf{F}_\text{V}'$, while the impact of frame length is opposite. The ISD and target velocity $|\mathbf{v}|$ are proportional to $\Bar{r}$ (or $\Bar{r}^\text{t}$ and $\Bar{r}^\text{r}$) and $\Bar{v}$, respectively strengthening and weakening the validity of the criterion in \textbf{Theorem~\ref{t.1}}.


\section{Analysis and Discussion}\label{sec.dis}
The simulations in Section~\ref{sec.perf} investigates the performance of the network-based cooperative sensing system via the full CRLB. The results show that from the perspective of accuracy and coverage, the sensing types in Fig.~\ref{fig.heatmap} can be ranked as 
\begin{align}
    4\negthinspace\times\negthinspace4\,\text{MXS}>4\negthinspace\times\negthinspace3\,\text{MBS}>4\,\text{MMS}>1\negthinspace\times\negthinspace4\,\text{MXS}>1\negthinspace\times\negthinspace3\,\text{MBS},\negthinspace\nonumber
\end{align}
implying that more measurements and better spatial diversity increase sensing performance.

Afterward, the validity of the simplified CRLB is examined in Section~\ref{sec.dev}, where the tightness of the proposed criterion for the applicability of the simplified CRLB is also investigated. The results show that the simplified CRLB can be safely adopted for the square enclosed by the sensing nodes, while the positions near the square edge and the nodes generally suffer from a relatively high simplification error.
The proposed criterion is particularly tight for bistatic sensing links since the high NMAEs near the edges and nodes are captured, while monostatic links can degrade the tightness. However, for all the sensing types, the minor simplification error in the center area can always be indicated by the proposed criterion.

Moreover, the impact of resource allocation and system parameters on the full CRLBs and the validity of the simplified CRLBs is evaluated. Their effects are explained in the following with reference to (\ref{eq.efimmono}) and (\ref{eq.efimbi}):
\begin{itemize}
    \item Power $P_\text{t}$ or ESNR $\gamma$: Since the signal-level EFIM components $E_r,E_{v_\text{r}},E_\theta\propto P_\text{t}, \gamma$, increasing $P_\text{t}$ or $\gamma$ can lead to simultaneously and proportionally enlarged full EFIMs $\mathbf{F}_\text{P}',\mathbf{F}_\text{V}'$ and simplified ones $\mathbf{F}_{\text{P}0},\mathbf{F}_{\text{V}0}$, as well as the coupling terms $\pmb{\Delta}_\text{P},\pmb{\Delta}_\text{V}$. Therefore, the PEB and VEB are reduced, while $\delta_\text{P}$ and $\delta_\text{V}$ are not changed.
    \item Receiver antenna horizontal size $N_\text{r}$: $E_r,E_{v_\text{r}}\propto N_\text{r}$, $E_\theta\propto N_\text{r}^3$. However, as analyzed in Section~\ref{sec.simpmms}, the contribution of $E_\theta$ is much less than $E_r$ and $E_{v_\text{r}}$ for the considered system configuration. In this case, the impact of $N_\text{r}$ is generally similar to that of ESNR, and the influence on $\delta_\text{P}$ and $\delta_\text{V}$ is tiny.
    \item Bandwidth $B$: Since $E_r\propto B^2$, the full and simplified PEBs, as well as $\pmb{\Delta}_\text{P}$, are inversely related to $B$. Consequently, the VEB also decreases slightly as $B$ increases. A wider bandwidth results in a smaller proportion of $\pmb{\Delta}_\text{V}$ in $\mathbf{F}_\text{P}'$ and $\pmb{\Delta}_\text{P}$ in $\mathbf{F}_\text{V}'$, hence reducing $\delta_\text{P}$ and $\delta_\text{V}$.
    \item Frame length $T_\text{F}$: Similar to $B$, a longer $T_\text{F}$ suppresses full and simplified VEBs and slightly reduces the full PEB through $\pmb{\Delta}_\text{V}$. Nevertheless, since $E_{v_\text{r}},\pmb{\Delta}_\text{V}\propto T_\text{F}^2$, the gap between the full and simplified PEBs increases with $T_\text{F}$. The same conclusion applies to the gap in VEBs since $\mathbf{F}_\text{V}'$ can be written in the form of $\mathbf{F}_\text{V}'=T_\text{F}^2\mathbf{A}-{||\mathbf{v}||^2T_\text{F}^4\mathbf{B}}({||\mathbf{v}||^2T_\text{F}^2\mathbf{D}+\mathbf{G}})^{-1}$, where $\mathbf{A},\mathbf{B},\mathbf{D}$, and $\mathbf{G}$ are the $||\mathbf{v}||,T_\text{F}$-irrelevant matrices. The relative contribution of $\pmb{\Delta}_\text{P}={||\mathbf{v}||^2T_\text{F}^4\mathbf{B}}({||\mathbf{v}||^2T_\text{F}^2\mathbf{D}+\mathbf{G}})^{-1}$ in $\mathbf{F}_\text{V}'$ increases with $T_\text{F}$, resulting in a larger approximation loss between the simplified and full VEBs.
    \item ISD: For a given position in the square formed by the four sensing nodes, since $\gamma\propto 1/(r^\text{t}r^\text{r})^2$ and $r^\text{t},r^\text{r}\propto \text{ISD}$, the PEB and VEB increase with ISD. Meanwhile, ${\Bar{r}^\text{t}\Bar{r}^\text{r}}/\Bar{v}^2$ increases, strengthening the criteria in (\ref{eq.crinode}) and (\ref{eq.crinodeb}), implying a decreasing trend of $\delta_\text{P}$ and $\delta_\text{V}$.
    \item Target velocity $\mathbf{v}$: Since $\pmb{\Delta}_\text{V}\propto||\mathbf{v}||^2$, the velocity also contributes to position information, leading to a lower full PEB but also a higher $\delta_\text{P}$, while the simplified PEB is not influenced. At the same time, $\pmb{\Delta}_\text{P}$ increases with $||\mathbf{v}||^2$, although not so significantly as $\pmb{\Delta}_\text{V}$. Consequently, the full VEB is increased with velocity, implying that higher velocity may degrade velocity estimation. The gap between the full and simplified VEBs is also enlarged.
\end{itemize}

\section{Conclusion}\label{sec.con}
This paper provides a comprehensive \acf{CRLB} analysis for cooperative \acf{MIMO} \acf{ISAC} networks. Starting from the signal-level parameters, we show the transformation to the \acf{EFIM} of the state parameters and derive the closed-form expressions for the full \acf{PEB} and \acf{VEB}, taking the mutual influence between position and velocity information into account. Particularly, the CRLBs of various sensing types, including \acf{MMS}, \acf{MBS}, and \acf{MXS}, are discussed respectively. Addressing the intractability and complexity of the full CRLB, we simplify the CRLB by excluding the coupling terms and provide a criterion for the validity of the simplified CRLBs. The theoretical analysis demonstrates that the criterion can be generally satisfied in representative 5G scenarios.
The simulation examines the coverage and accuracy of different cooperative sensing types and the validity of the simplified CRLB in heatmaps. 
The results demonstrate that the simplified CRLB can be safely adopted for the polygon enclosed by the sensing nodes, which is generally the main cooperative sensing area. However, a relatively higher error is observed at the polygon edges and near the sensing nodes. 
The sufficiency and tightness of the proposed criterion are also demonstrated. 
In addition, the impact of system configuration and state parameters on the full CRLB and the validity of the simplified CRLB is evaluated and explained with reference to their formulas.

In summary, while the full CRLB correctly reflects the maximum achievable accuracy of a sensing network and guides geometric deployment, the simplified CRLB provides a tractable sensing performance metric for optimization problems such as radio resource allocation and beamforming, accelerating the integration of sensing functions into future mobile networks. Our future work will leverage the full and simplified CRLBs for beamforming in cooperative ISAC networks and quantify the performance loss of the simplified CRLBs against their gains in computational efficiency, flexibility, and real-time applicability.

\printbibliography

@ARTICLE{10536135,
  author={González-Prelcic, Nuria and Furkan Keskin, Musa and Kaltiokallio, Ossi and Valkama, Mikko and Dardari, Davide and Shen, Xiao and Shen, Yuan and Bayraktar, Murat and Wymeersch, Henk},
  journal={Proceedings of the IEEE}, 
  title={The Integrated Sensing and Communication Revolution for 6G: Vision, Techniques, and Applications}, 
  year={2024},
  volume={112},
  number={7},
  pages={676-723},
  doi={10.1109/JPROC.2024.3397609}}

@misc{persp,
title={German Perspective on 6G - Use Cases, Technical Building Blocks and Requirements}, 
DOI={10.24406/publica-4506}, 
author={Franchi, Norman and Dressler, Falko}, 
year={2024},
month = {Dec},
series = {FAU Studien aus der Elektrotechnik},
  volume = {28},
  publisher = {FAU University Press},
  ISBN = {978-3-96147-797-5},
address={Erlangen, Germany},
booktitle={FAU Studien aus der Elektrotechnik},
}

@ARTICLE{9924202,
  author={Zhou, Wenxing and Zhang, Ruoyu and Chen, Guangyi and Wu, Wen},
  journal={IEEE Open Journal of the Communications Society}, 
  title={Integrated Sensing and Communication Waveform Design: A Survey}, 
  year={2022},
  volume={3},
  number={},
  pages={1930-1949},
  doi={10.1109/OJCOMS.2022.3215683}}

@ARTICLE{9737357,
  author={Liu, Fan and Cui, Yuanhao and Masouros, Christos and Xu, Jie and Han, Tony Xiao and Eldar, Yonina C. and Buzzi, Stefano},
  journal={IEEE Journal on Selected Areas in Communications}, 
  title={Integrated Sensing and Communications: Toward Dual-Functional Wireless Networks for 6G and Beyond}, 
  year={2022},
  volume={40},
  number={6},
  pages={1728-1767},
  doi={10.1109/JSAC.2022.3156632}}

@ARTICLE{9354629,
  author={Wild, Thorsten and Braun, Volker and Viswanathan, Harish},
  journal={IEEE Access}, 
  title={Joint Design of Communication and Sensing for Beyond 5G and 6G Systems}, 
  year={2021},
  volume={9},
  number={},
  pages={30845-30857},
  doi={10.1109/ACCESS.2021.3059488}}

@ARTICLE{9591277,
  author={Nidamanuri, Jaswanth and Nibhanupudi, Chinmayi and Assfalg, Rolf and Venkataraman, Hrishikesh},
  journal={IEEE Transactions on Intelligent Vehicles}, 
  title={A Progressive Review: Emerging Technologies for ADAS Driven Solutions}, 
  year={2022},
  volume={7},
  number={2},
  pages={326-341},
  doi={10.1109/TIV.2021.3122898}}

@INPROCEEDINGS{9272226,
  author={Benhaddou, Driss and Harikumar, Mayuri and Chen, Ji and Han, Zhu and Bhaskar, Vidhyacharan and Abdelhadi, Ahmed},
  booktitle={2020 IEEE International Symposium on Systems Engineering (ISSE)}, 
  title={Coordinated Multipoint User Scheduling for 5G Cloud Radio Access Network}, 
  year={2020},
  location={Vienna, Austria},
  volume={},
  number={},
  pages={1-7},
  doi={10.1109/ISSE49799.2020.9272226}}

@ARTICLE{9762865,
  author={Ranjbar, Vida and Girycki, Adam and Rahman, Md Arifur and Pollin, Sofie and Moonen, Marc and Vinogradov, Evgenii},
  journal={IEEE Communications Standards Magazine}, 
  title={Cell-Free mMIMO Support in the O-RAN Architecture: A PHY Layer Perspective for 5G and Beyond Networks}, 
  year={2022},
  volume={6},
  number={1},
  pages={28-34},
  doi={10.1109/MCOMSTD.0001.2100067}}

@ARTICLE{9585321,
  author={Zhang, J. Andrew and Rahman, Md. Lushanur and Wu, Kai and Huang, Xiaojing and Guo, Y. Jay and Chen, Shanzhi and Yuan, Jinhong},
  journal={IEEE Communications Surveys \& Tutorials}, 
  title={Enabling Joint Communication and Radar Sensing in Mobile Networks—A Survey}, 
  year={2022},
  volume={24},
  number={1},
  pages={306-345},
  doi={10.1109/COMST.2021.3122519}}

@ARTICLE{10599241,
  author={Su, Yanpeng and Lübke, Maximilian and Franchi, Norman},
  journal={IEEE Access}, 
  title={Coordinated Multipoint JCAS in 6G Mobile Networks}, 
  year={2024},
  volume={12},
  number={},
  pages={98530-98545},
  doi={10.1109/ACCESS.2024.3428831}}

@ARTICLE{1420803,
  author={Smith, S.T.},
  journal={IEEE Transactions on Signal Processing}, 
  title={Statistical resolution limits and the complexified Crame/spl acute/r-Rao bound}, 
  year={2005},
  volume={53},
  number={5},
  pages={1597-1609},
  doi={10.1109/TSP.2005.845426}}

@ARTICLE{923295,
  author={Dogandzic, A. and Nehorai, A.},
  journal={IEEE Transactions on Signal Processing}, 
  title={Cramer-Rao bounds for estimating range, velocity, and direction with an active array}, 
  year={2001},
  volume={49},
  number={6},
  pages={1122-1137},
  doi={10.1109/78.923295}}

@ARTICLE{7347470,
  author={Zhao, Tong and Huang, Tianyao},
  journal={IEEE Transactions on Signal Processing}, 
  title={Cramer-Rao Lower Bounds for the Joint Delay-Doppler Estimation of an Extended Target}, 
  year={2016},
  volume={64},
  number={6},
  pages={1562-1573},
  doi={10.1109/TSP.2015.2505681}}

@ARTICLE{9627227,
  author={Giroto de Oliveira, Lucas and Nuss, Benjamin and Alabd, Mohamad Basim and Diewald, Axel and Pauli, Mario and Zwick, Thomas},
  journal={IEEE Transactions on Microwave Theory and Techniques}, 
  title={Joint Radar-Communication Systems: Modulation Schemes and System Design}, 
  year={2022},
  volume={70},
  number={3},
  pages={1521-1551},
  doi={10.1109/TMTT.2021.3126887}}

@INPROCEEDINGS{7131233,
  author={Ai, Yue and Yi, Wei and Blum, Rick S. and Kong, Lingjiang},
  booktitle={2015 IEEE Radar Conference (RadarCon)}, 
  title={Cramer-Rao lower bound for multitarget localization with noncoherent statistical MIMO radar}, 
  year={2015},
  location={Arlington, VA, USA},
  volume={},
  number={},
  pages={1497-1502},
  doi={10.1109/RADAR.2015.7131233}}

@INPROCEEDINGS{10615966,
  author={Zabini, Flavio and Paolini, Enrico and Xu, Wen and Giorgetti, Andrea},
  booktitle={2024 IEEE International Conference on Communications Workshops (ICC Workshops)}, 
  title={Fundamental Limits of Cooperative Strategies in Joint Sensing and Communication Networks}, 
  year={2024},
  location={Denver, CO, USA},
  volume={},
  number={},
  pages={329-334},
  doi={10.1109/ICCWorkshops59551.2024.10615966}}

@INPROCEEDINGS{11202824,
  author={Zhao, Zhixiang and Semper, Sebastian and Schneider, Christian and Thomä, Reiner S.},
  booktitle={2025 28th International Workshop on Smart Antennas (WSA)}, 
  title={Sensing-Aided Beamforming: The Impact of Distributed Sensing Network Geometry}, 
  year={2025},
  location={Erlangen, Germany},
  volume={},
  number={},
  pages={1-7},
  doi={10.1109/WSA65299.2025.11202824}}

@INPROCEEDINGS{10646234,
  author={Mollén, Christopher and Huschke, Jörg and Baldemair, Robert and Fodor, Gabor},
  booktitle={2024 IEEE 4th International Symposium on Joint Communications \& Sensing (JC\&S)}, 
  title={Multistatic Sensing Performance Maps for Evaluating Integrated Sensing and Communication Deployments}, 
  year={2024},
  location={Leuven, Belgium},
  volume={},
  number={},
  pages={1-6},
  doi={10.1109/JCS61227.2024.10646234}}

@INPROCEEDINGS{10419729,
  author={Yao, SiNuo and Zhang, JunTao and Cai, Shu and Zhang, Jun},
  booktitle={2023 IEEE 23rd International Conference on Communication Technology (ICCT)}, 
  title={A Performance Bound for Target Localization in an OFDM-Based Communication Network}, 
  year={2023},
  location={Wuxi, China},
  volume={},
  number={},
  pages={479-483},
  doi={10.1109/ICCT59356.2023.10419729}}

@INPROCEEDINGS{10693955,
  author={Pucci, Lorenzo and Giorgetti, Andrea},
  booktitle={2024 IEEE 25th International Workshop on Signal Processing Advances in Wireless Communications (SPAWC)}, 
  title={Position Error Bound for Cooperative Sensing in MIMO-OFDM Networks}, 
  year={2024},
  location={Lucca, Italy},
  volume={},
  number={},
  pages={296-300},
  doi={10.1109/SPAWC60668.2024.10693955}}

@INPROCEEDINGS{10266619,
  author={Xu, Yiming and Xie, Lei and Xu, Dongfang and Song, Shenghui},
  booktitle={2023 IEEE International Mediterranean Conference on Communications and Networking (MeditCom)}, 
  title={Fundamental Limits and Base Station Selection for Collaborative Sensing in Perceptive Mobile Networks}, 
  year={2023},
  location={Dubrovnik, Croatia},
  volume={},
  number={},
  pages={97-102},
  doi={10.1109/MeditCom58224.2023.10266619}}

@ARTICLE{7362229,
  author={He, Qian and Hu, Jianbin and Blum, Rick S. and Wu, Yonggang},
  journal={IEEE Transactions on Signal Processing}, 
  title={Generalized Cramér–Rao Bound for Joint Estimation of Target Position and Velocity for Active and Passive Radar Networks}, 
  year={2016},
  volume={64},
  number={8},
  pages={2078-2089},
  doi={10.1109/TSP.2015.2510978}}

@ARTICLE{11481148,
  author={Xia, Guoqing and Xiao, Pei and Luo, Qu and Ji, Bing and Zhang, Yue and Zhou, Huiyu},
  journal={IEEE Transactions on Communications}, 
  title={Joint Location and Velocity Estimation and Fundamental CRLB Analysis for Cell-Free MIMO-ISAC}, 
  year={2026},
  volume={74},
  number={},
  pages={7357-7374},
  doi={10.1109/TCOMM.2026.3683820}}

@ARTICLE{11231051,
  author={Arcangeloni, Luca and Testi, Enrico and Pucci, Lorenzo and Giorgetti, Andrea},
  journal={IEEE Open Journal of the Communications Society}, 
  title={Fundamental Limits of Target Parameter Estimation in OFDM-Based 3D NTN ISAC Systems}, 
  year={2025},
  volume={6},
  number={},
  pages={9534-9546},
  doi={10.1109/OJCOMS.2025.3630072}}

@misc{pucci2025,
      title={Position and Velocity Estimation Accuracy in MIMO-OFDM ISAC Networks: A Fisher Information Analysis}, 
      author={Lorenzo Pucci and Luca Arcangeloni and Andrea Giorgetti},
      year={2025},
      eprint={2507.01743},
      archivePrefix={arXiv},
      primaryClass={eess.SP}, 
}

@misc{xu2026,
      title={Hybrid Mono- and Bi-static OFDM-ISAC via BS-UE Cooperation: Closed-Form CRLB and Coverage Analysis}, 
      author={Xiaoli Xu and Yong Zeng},
      year={2026},
      eprint={2601.09057},
      archivePrefix={arXiv},
      primaryClass={cs.IT},
}

@ARTICLE{5393291,
  author={He, Qian and Blum, Rick S. and Godrich, Hana and Haimovich, Alexander M.},
  journal={IEEE Journal of Selected Topics in Signal Processing}, 
  title={Target Velocity Estimation and Antenna Placement for MIMO Radar With Widely Separated Antennas}, 
  year={2010},
  volume={4},
  number={1},
  pages={79-100},
  doi={10.1109/JSTSP.2009.2038974}}

@ARTICLE{8869705,
  author={Saad, Walid and Bennis, Mehdi and Chen, Mingzhe},
  journal={IEEE Network}, 
  title={A Vision of 6G Wireless Systems: Applications, Trends, Technologies, and Open Research Problems}, 
  year={2020},
  volume={34},
  number={3},
  pages={134-142},
  doi={10.1109/MNET.001.1900287}}

@techreport{3gpp38913,
  author      = {{3GPP}},
  title       = {{Study on Scenarios and Requirements for Next Generation Access Technologies}},
  institution = {{3rd Generation Partnership Project (3GPP)}},
  number      = {TR 38.913 V19.0.0},
  type        = {Technical Report},
  year        = {2025},
  note        = {Release 19},
  url         = {https://www.3gpp.org/ftp/Specs/archive/38_series/38.913/38913-j00.zip}
}

@techreport{3gpp.36.819,
 author = {3GPP},
 institution = {{3rd Generation Partnership Project (3GPP)}},
 month = {05},
 note = {Version 1.0.0},
 number = {36.819},
 title = {Coordinated Multi-Point Operation for LTE Physical Layer Aspects},
 type = {Technical Specification (TS)},
 url = {https://portal.3gpp.org/desktopmodules/Specifications/SpecificationDetails.aspx?specificationId=2498},
 year = {2011}
}

@INPROCEEDINGS{9411178,
  author={Thomä, Reiner and Dallmann, Thomas and Jovanoska, Snezhana and Knott, Peter and Schmeink, Anke},
  booktitle={2021 15th European Conference on Antennas and Propagation (EuCAP)}, 
  title={Joint Communication and Radar Sensing: An Overview}, 
  year={2021},
  location={Dusseldorf, Germany},
  volume={},
  number={},
  pages={1-5},
  doi={10.23919/EuCAP51087.2021.9411178}}

@ARTICLE{10577579,
  author={Babu, Nithin and Masouros, Christos and Papadias, Constantinos B. and Eldar, Yonina C.},
  journal={IEEE Transactions on Wireless Communications}, 
  title={Precoding for Multi-Cell ISAC: From Coordinated Beamforming to Coordinated Multipoint and Bi-Static Sensing}, 
  year={2024},
  volume={23},
  number={10},
  pages={14637-14651},
  doi={10.1109/TWC.2024.3417713}}

@ARTICLE{9945983,
  author={Dong, Fuwang and Liu, Fan and Cui, Yuanhao and Wang, Wei and Han, Kaifeng and Wang, Zhiqin},
  journal={IEEE Transactions on Wireless Communications}, 
  title={Sensing as a Service in 6G Perceptive Networks: A Unified Framework for ISAC Resource Allocation}, 
  year={2023},
  volume={22},
  number={5},
  pages={3522-3536},
  doi={10.1109/TWC.2022.3219463}}

@ARTICLE{8356190,
  author={Abu-Shaban, Zohair and Zhou, Xiangyun and Abhayapala, Thushara and Seco-Granados, Gonzalo and Wymeersch, Henk},
  journal={IEEE Transactions on Wireless Communications}, 
  title={Error Bounds for Uplink and Downlink 3D Localization in 5G Millimeter Wave Systems}, 
  year={2018},
  volume={17},
  number={8},
  pages={4939-4954},
  doi={10.1109/TWC.2018.2832134}}

@ARTICLE{4408448,
  author={Haimovich, Alexander M. and Blum, Rick S. and Cimini, Leonard J.},
  journal={IEEE Signal Processing Magazine}, 
  title={MIMO Radar with Widely Separated Antennas}, 
  year={2008},
  volume={25},
  number={1},
  pages={116-129},
  doi={10.1109/MSP.2008.4408448}}

@ARTICLE{11570946,
  author={Su, Yanpeng and Franchi, Norman and Lübke, Maximilian},
  journal={IEEE Transactions on Radar Systems}, 
  title={On Unified CRLB Framework from Generic Signals to ISAC Waveforms with Virtual Array Sensing}, 
  year={2026},
  volume={},
  number={},
  pages={1-1},
  doi={10.1109/TRS.2026.3705154}}

\end{document}